\documentclass[preprint]{imsart}

\usepackage{lineno,hyperref}

\usepackage{aliascnt}
\usepackage{amsmath}
\usepackage{amssymb}
\usepackage{amsthm}
\usepackage{bibentry}
\usepackage{bbm}
\usepackage{booktabs}
\usepackage{color}
\usepackage{enumerate}
\usepackage{float}
\usepackage{graphicx}
\usepackage{hyperref}
\usepackage{ifthen}
\usepackage{mathtools}
\usepackage[numbers]{natbib}
\usepackage[linesnumbered,vlined]{algorithm2e}
\usepackage{stmaryrd}
\usepackage{ushort}
\usepackage{xargs}

\RequirePackage[colorlinks,citecolor=blue,urlcolor=blue]{hyperref}

\startlocaldefs

\newtheorem{theorem}{Theorem}
\newaliascnt{proposition}{theorem}
\newtheorem{proposition}[proposition]{Proposition}
\aliascntresetthe{proposition}
\newaliascnt{lemma}{theorem}
\newtheorem{lemma}[lemma]{Lemma}
\aliascntresetthe{lemma}
\newaliascnt{definition}{theorem}
\newtheorem{definition}[definition]{Definition}
\aliascntresetthe{definition}
\newaliascnt{example}{theorem}
\newtheorem{example}[example]{Example}
\aliascntresetthe{example}
\newaliascnt{remark}{theorem}
\newtheorem{remark}[remark]{Remark}
\aliascntresetthe{remark}

\newcommandx\A[2][1=]{
\ifthenelse{\equal{#1}{}}
{\hspace{-1mm}(\textbf{A\ref{#2}})\hspace{-1mm}}
{\hspace{-1mm}(\textbf{A\ref{#1}--\ref{#2}})\hspace{-1mm}}
}

\newcommandx\B[2][1=]{
\ifthenelse{\equal{#1}{}}
{\hspace{-1mm}(\textbf{S})\hspace{-1mm}}
{\hspace{-1mm}(\textbf{S\ref{#1}--\ref{#2}})\hspace{-1mm}}
}
\newcommand{\bd}{c}
\newcommand{\bias}[2]{\beta_{#2} \langle #1 \rangle} 
\newcommand{\biasfilt}[2]{\bar{\beta}_{#2} \langle #1 \rangle} 
\newcommand{\binset}[1]{\mathsf{I}_{#1}}
\newcommand{\Binsp}[1]{\mathsf{B}_{#1}}
\newcommand{\bmf}[1]{\mathbb{F}(#1)}

\newcommand{\cat}{\mathsf{Cat}}
\newcommand{\chunk}[3]{{#1}_{#2}^{#3}}

\newcommand{\DDelta}[3]{\Delta_{#1}\langle #2\rangle(#3)}
\newcommand{\dlim}{\stackrel{\mathcal D}{\longrightarrow}}

\newcommand{\E}{\mathbb{E}}
\newcommand{\ed}{g}
\newcommand{\Efd}{\mathcal{E}}
\newcommand{\enoch}[3]{E_{#1,#2}^{#3}}
\newcommand{\epart}[2]{\xi_{#1}^{#2}}
\newcommand{\eqdef}{\vcentcolon=}
\newcommand{\Esp}{\mathsf{E}}
\newcommandx{\eve}[3][1=]{\ifthenelse{\equal{#1}{}}{E_{#2}^{#3}}{E_{#1,#2}^{#3}}}

\newcommandx{\filt}[1][1=]{\ifthenelse{\equal{#1}{}}{\filtsymb}{\filtsymb \langle #1 \rangle}}
\newcommand{\filtpart}[1][1=]{\ifthenelse{\equal{#1}{}}{\filtsymb_\N}{\filtsymb_\N \langle #1 \rangle}}
\newcommand{\filtsymb}{\bar{\eta}}
\newcommandx{\filtvariance}[3][1=,3=]{\bar{\sigma}^{#1}_{#3} \langle #2 \rangle}

\newcommandx{\gen}[1][1=]{\ifthenelse{\equal{#1}{1}}{G}{G'}} 
\newcommand{\genkernel}{\kernel{K}}

\newcommand{\hk}{\kernel{M}}

\newcommand{\init}{\chi}
\newcommand{\ind}[2]{I_{#1}^{#2}}
\newcommand{\intvect}[2]{\llbracket #1, #2 \rrbracket}

\newcommand{\kernel}[1]{\mathbf{#1}}

\newcommand{\lag}{\lambda}
\newcommand{\lagtime}[2]{#1(#2)}
\newcommandx{\likeli}[3][1=]{\pi_{#1} \langle #2 \rangle(#3)}
\newcommand{\limitfunc}[1]{\pi \langle #1 \rangle}

\newcommand{\mdlow}{\ushort{\varepsilon}}
\newcommand{\mdup}{\bar{\varepsilon}}
\newcommand{\mdr}{\mathbb{M}}
\newcommand{\me}{\mathrm{e}}
\newcommand{\mk}{\kernel{M}}

\newcommand{\mumeas}[2]{\mu_{#1} \langle #2 \rangle}

\newcommand{\N}{N}
\newcommand{\nset}{\mathbb{N}}
\newcommand{\nsetpos}{\mathbb{N}^\ast}

\newcommand{\1}{\mathbbm{1}}
\newcommand{\ordo}{\mathcal{O}}

\newcommand{\partfd}[1]{\mathcal{F}_{#1}}
\newcommand{\per}{\zeta}
\newcommand{\perblock}{\bar{\zeta}}
\newcommand{\plim}{\stackrel{\prob}{\longrightarrow}}
\newcommandx{\pot}[1][1=]{\ifthenelse{\equal{#1}{}}{g}{g \langle #1 \rangle}}
\newcommand{\potlow}{\ushort{\delta}}
\newcommand{\potup}{\bar{\delta}}
\newcommand{\predsymb}{\eta}
\newcommandx{\pred}[1][1=]{\ifthenelse{\equal{#1}{}}{\predsymb}{\predsymb \langle #1 \rangle}}
\newcommand{\predpart}[1][1=]{\ifthenelse{\equal{#1}{}}{\predsymb_\N}{\predsymb_\N \langle #1 \rangle}}
\newcommand{\prob}{\mathbb{P}}
\newcommand{\probmeas}[1]{\mathbb{M}(#1)}
\newcommand{\probdoeblin}[2]{\mu_{#1} \langle #2 \rangle}
\newcommand{\rmd}{\mathrm{d}}

\newcommand{\rate}{\rho}
\newcommand{\refm}{\nu}
\newcommand{\rset}{\mathbb{R}}
\newcommand{\rsetpos}{\mathbb{R}^\ast_+}

\newcommand{\term}[3][]{\upsilon_{#2,#3} \langle #1 \rangle}

\newcommandx{\uk}[1][1=]{\ifthenelse{\equal{#1}{}}{\kernel{Q}}{\kernel{Q} \langle #1 \rangle}}
\newcommand{\unitstr}[2]{1_{#1}}
\newcommand{\unpredsymbol}{\gamma}
\newcommandx{\unpred}[1][1=]{\ifthenelse{\equal{#1}{}}{\unpredsymb}{\unpredsymb \langle #1 \rangle}}
\newcommand{\unpredpart}[1][1=]{\ifthenelse{\equal{#1}{}}{\unpredsymbol_\N}{\unpredsymbol_\N \langle #1 \rangle}}

\newcommandx{\varest}[3][1=,3=]{\ifthenelse{\equal{#3}{}}{\sigma^{#1}_\N \langle #2 \rangle}{\sigma^{#1}_{\N, #3} \langle #2 \rangle}}
\newcommandx{\varestfilt}[3][1=,3=]{\ifthenelse{\equal{#3}{}}{\bar{\sigma}^{#1}_\N \langle #2 \rangle}{\bar{\sigma}^{#1}_{\N, #3} \langle #2 \rangle}}
\newcommandx{\variance}[3][1=,3=]{\sigma^{#1}_{#3} \langle #2 \rangle}

\newcommand{\wgt}[2]{\omega_{#1}^{#2}}
\newcommand{\wgtsum}[1]{\Omega_{#1}}

\newcommand{\Xsp}{\mathsf{X}}
\newcommand{\Xfd}{\mathcal{X}}

\newcommand{\Ysp}{\mathsf{Y}}
\newcommand{\Yfd}{\mathcal{Y}}

\newcommand{\zerostr}[1]{0_{#1}}
\newcommand{\Zsp}{\mathsf{Z}}
\newcommand{\Zfd}{\mathcal{Z}}
\newcommand{\zset}{\mathbb{Z}}

\newcounter{hypA}
\newenvironment{hypA}{\refstepcounter{hypA}\begin{itemize}
  \item[({\bf A\arabic{hypA}})]}{\end{itemize}}

\newcounter{hypB}
\newenvironment{hypB}{\refstepcounter{hypB}\begin{itemize}
  \item[({\bf S})]}{\end{itemize}}

\endlocaldefs

\begin{document}

\begin{frontmatter}


\title{Numerically stable online estimation of variance in particle filters}

\runtitle{Estimation of variance in particle filters}

\begin{aug}
\author{\fnms{Jimmy} \snm{Olsson}\thanksref{a,t1}\ead[label=e1]{jimmyol@kth.se}}
\and
\author{\fnms{Randal} \snm{Douc}\thanksref{b}\ead[label=e2]{randal.douc@it-sudparis.eu}}

\address[a]{Department of Mathematics \\
KTH Royal Institute of Technology \\
SE-100 44  Stockholm, Sweden \\
\printead{e1}}

\address[b]{D\'epartement CITI \\ 
TELECOM SudParis \\
9 rue Charles Fourier, 91000 EVRY \\
\printead{e2}}

\thankstext{t1}{J.~Olsson is supported by the Swedish Research Council, Grant 2011-5577.}

\runauthor{J. Olsson and R. Douc}

\affiliation{KTH Royal Institute of Technology and Institut T\'el\'ecom/T\'el\'ecom SudParis}

\end{aug}

\begin{abstract}
This paper discusses variance estimation in sequential Monte Carlo methods, alternatively termed particle filters. The variance estimator that we propose is a natural modification of that suggested by H.~P.~Chan and T.~L.~Lai [A general theory of particle filters in hidden Markov models and some applications. \emph{Ann. Statist.}, 41(6):2877--2904, 2013], which allows the variance to be estimated in a single run of the particle filter by tracing the genealogical history of the particles. However, due particle lineage degeneracy, the estimator of the mentioned work becomes numerically unstable as the number of sequential particle updates increases. Thus, by tracing only a part of the particles' genealogy rather than the full one, our estimator gains long-term numerical stability at the cost of a bias. The scope of the genealogical tracing is regulated by a lag, and under mild, easily checked model assumptions, we prove that the bias tends to zero geometrically fast as the lag increases. As confirmed by our numerical results, this allows the bias to be tightly controlled also for moderate particle sample sizes. 
\end{abstract}





\end{frontmatter}

\section{Introduction}
\label{sec:introduction}
Since the \emph{bootstrap particle filter} was introduced in \cite{gordon:salmond:smith:1993}, \emph{sequential Monte Carlo} (SMC) \emph{methods}, alternatively termed \emph{particle filters}, have been successfully applied within a wide range of applications, including computer vision, automatic control, signal processing, optimisation, robotics, econometrics, and finance; see, e.g.,  \cite{doucet:defreitas:gordon:2001,ristic:arulampalam:gordon:2004} for introductions to the topic. SMC methods approximate a given sequence of distributions by a sequence of possibly weighted empirical measures associated with a sample of \emph{particles} evolving recursively and randomly in time. Each iteration of the SMC algorithm comprises two operations: a \emph{mutation step}, which moves the particles randomly in the state space, and a \emph{selection step}, which duplicates/eliminates, through resampling, particles with high/low importance weights, respectively.  

In parallel with algorithmic developments, the theoretical properties of SMC have been studied extensively during the last twenty years, and there is currently a number of available results describing the convergence, as the number of particles tends to infinity, of Monte Carlo estimates produced by the algorithm; see, e.g., the monographs \cite{delmoral:2004,delmoral:2013} and \cite[Chapter~9]{cappe:moulines:ryden:2005}. The first \emph{central limit theorem} (CLT) for SMC methods was established in \cite{delmoral:guionnet:1999}, and this result was later refined in the series of papers \cite{chopin:2004,kuensch:2005,douc:moulines:2008}. In the mentioned CLT, the asymptotic variance of the weak Gaussian limit is expressed through a recursive formula involving high-dimensional integrals over generally complicated integrands and is hence intractable in general. 

Due to its complexity, only a very few recent works have treated the important---although challenging---topic of variance estimation in SMC algorithms. A breakthrough was made by H.~P. Chan and T.~L.~Lai, who proposed, in \cite{chan:lai:2013} and within the framework of general \emph{state-space models} (or, \emph{hidden Markov models}), an estimator, from now on referred to as the \emph{Chan \& Lai estimator} (CLE), that allows the sequence of asymptotic variances to be estimated on the basis of a \emph{single} realisation of the algorithm, without the need of additional simulations. Remarkably, the CLE can be shown to be consistent (see \cite[Theorem~2]{chan:lai:2013}), i.e., to converge in probability to the true asymptotic variance as the number of particles tends to infinity. At a given time step, the CLE estimates the asymptotic variance by tracing genealogically the time-zero ancestors of the particles at the time step in question. The variance estimators proposed recently in \cite{lee:whiteley:2016} are based on the same principle, and may be viewed as refinements of the CLE within a more general framework of \emph{Feynman-Kac models} and particle algorithms with time varying particle population sizes. Moreover, \cite{lee:whiteley:2016} provides an elegant, deepened  (asymptotic as well as non-asymptotic) theoretical analysis of the technique, and these results are essential for the development of the present paper. 

Appealingly, the set of time-zero ancestors may be updated recursively in the SMC algorithm by adding  just a single line to the code. This allows variance estimates to be computed online with essentially the same computational complexity and memory demands as the original algorithm. Nevertheless, since the SMC algorithm performs repeatedly selection, it is well established that all the particles in the sample will, eventually, share the \emph{same} time-zero ancestor (see, e.g., \cite{jacob:murray:rubenthaler:2015} for a theoretical analysis of this particle path degeneracy phenomenon). Unfortunately, this implies that the CLE collapses eventually to zero as time increases. Increasing the particle sample size or resampling less frequently the particles will postpone somewhat, but not avoid, this collapse. This makes the CLE impractical in the long run. Thus, although the SMC methodology has become a standard tool in statistics and engineering, a numerically stable estimator of the SMC variance has, surprisingly, hitherto been lacking, and the aim of the present paper is to---at least partially---fill this gap.

The natural solution that we propose in the present paper is to estimate the SMC asymptotic variance by tracing only a \emph{part} of the particles' genealogy rather than the full one (i.e., back to the time-zero ancestors). By tracing genealogically only the last generations, depleted ancestor sets are avoided as long as the particle sample size is at least moderately large. Still, this measure leads to a bias, whose size determines completely the success of the approach. Nevertheless, in \cite{douc:moulines:olsson:2014}, which studies the stochastic stability of the sequence of asymptotic variances within the framework of general hidden Markov models, or, viewed differently, randomly perturbed Feynman-Kac models (see \cite[Remark~1]{douc:moulines:olsson:2014}), it is established that the variance, which at time $n$ may be expressed in the form of a sum of $n + 1$ terms, is, in the case where the perturbations form a stationary sequence, uniformly stochastically bounded in time, or, \emph{tight}. Moreover, the analysis provided in \cite{douc:moulines:olsson:2014}, which is driven by mixing assumptions, indicates that the size of the $m^{\mathrm{th}}$ term in the variance at time $n$ decreases geometrically fast with the difference $n - m$. Consequently, we may expect the last terms of the sum to represent the major part of the variance, and as long as the number $\lambda$ of traced generations, the \emph{lag}, and the particle sample size are not too small, the bias should be negligible. This argument is confirmed by our main result, Theorem~\ref{thm:tightness:bias}, which at any time point $n$ provides an order $\rho^\lag$ bound on the asymptotic (as the number of particles tends to infinity) bias, where $\rho \in (0, 1)$ is a mixing rate. Consequently, as long as the number of particles is large enough, the bias stays numerically stable in the long run and may, as it decreases geometrically fast with the lag $\lag$, be controlled efficiently. Methodologically, the estimator that we propose has similarities with the \emph{fixed-lag smoothing} approach studied in \cite{kitagawa:sato:2001,olsson:cappe:douc:moulines:2006,olsson:strojby:2010}, and we here face the same bias-variance tradeoff, in the sense that a too greedy/generous lag design leads to high bias/variance, respectively.  

The developments of the present paper are cast into the framework of randomly perturbed Feynman-Kac models, and the theoretical analysis is driven by the assumptions of \cite{douc:moulines:2012,douc:moulines:olsson:2014} (going back to \cite{douc:fort:moulines:priouret:2009}), which are easily checked for many models used in practice. In particular, by replacing the now classical strong mixing assumption on the transition kernel of the underlying Markov chain---a standard assumption in the literature that typically requires the state space of the Markov chain to be a compact set---by a \emph{local Doeblin condition}, we are able to verify the assumptions for a wide class of models with possibly non-compact state space. 

As a numerical illustration, we apply our estimator in the context of SMC-based predictor flow approximation in general state-space models, including the widely used \emph{stochastic volatility} model proposed in \cite{hull:white:1987}. We are able to report an efficient, numerically stable performance of our variance estimator, with tight control of the bias at low variance. 

The paper is structured as follows. Section~\ref{sec:preliminaries} introduces some notation, defines the framework of perturbed Feynman-Kac models, and provides some background to SMC. In Section~\ref{sec:estimator}, focus is set on asymptotic variance estimation and after a prefatory discussion on the CLE we introduce the proposed fixed-lag variance estimator. We also describe how our estimator can be straightforwardly extended to so-called \emph{updated} Feynman-Kac distribution flows. Theoretical and numerical results are found in Section~\ref{sec:theoretical:results} and Section~\ref{sec:numerical:study}, respectively, and Appendix~\ref{sec:proofs} contains all proofs. Our theoretical analysis is divided into two parts: first, the identification of the limiting bias and, second, the construction of a tight upper bound on the same. The first part relies on the theoretical machinery developed in \cite{lee:whiteley:2016}, whose key elements are recalled briefly in the beginning of Section~\ref{sec:proof:consistency:fixed:lag}. Finally, Section~\ref{sec:conclusion} concludes the paper.

\section{Preliminaries}
\label{sec:preliminaries}
\subsection{Some notation and conventions}

We assume that all random variables are defined on a common probability space $(\Omega, \mathcal{F}, \prob)$. The set of natural numbers is denoted by $\nset = \{0, 1, 2, \ldots\}$, and we let $\nsetpos = \nset \setminus \{Ê0\}$ be the positive ones. For all $(m, n) \in \nset^2$, we set $\intvect{m}{n} \eqdef \{m, m + 1, \ldots, n\}$. The set of nonnegative real numbers is denoted by $\rset_+$. For any quantities $\{ a_\ell \}_{\ell = 1}^m$, vectors are denoted by $\chunk{a}{\ell}{m} \eqdef (a_\ell, \ldots, a_m)$. 

We introduce some measure and kernel notation. Given some state space $(\Esp, \Efd)$, we denote by $\bmf{\Efd}$ and $\probmeas{\Efd}$ the spaces of bounded measurable functions and probability measures on $(\Esp, \Efd)$, respectively. For any functions $(h, h') \in \bmf{\Efd}^2$ we define the product function $h \varotimes h' : \Esp^2 \ni (x, x') \mapsto h(x) h'(x')$. The identity function $x \mapsto x$ is denoted by $\operatorname{id}$. Let $\mu$ be a measure on $(\Esp, \Efd)$; then for any $\mu$-integrable function $h$, we denote by
$$
\mu h \eqdef \int h(x) \, \mu(\rmd x)
$$
the Lebesgue integral of $h$ w.r.t. $\mu$. In addition, let $(\Esp', \Efd')$ be some other measurable space and $\genkernel$ some possibly unnormalised transition kernel $\genkernel : \Esp \times \Efd' \rightarrow \rset_+$. The kernel $\genkernel$ induces two integral operators, one acting on functions and the other on measures. More specifically, given a measure $\nu$ on $(\Esp, \Efd)$ and a measurable function $h$ on $(\Esp', \Efd')$, we define the measure 
$$
    \nu \genkernel : \Efd' \ni A \mapsto \int \genkernel(x, A) \, \nu(\rmd x) 
$$
and the function 
$$
    \genkernel h : \Esp \ni x \mapsto \int h(y) \, \genkernel(x, \rmd y),
$$
whenever these quantities are well defined.

\subsection{Randomly perturbed Feynman-Kac models} 
\label{sec:Feynman:Kac:models}
Let $(\Xsp, \Xfd)$ and $(\Zsp, \Zfd)$ be a pair of general measurable spaces. Moreover, let $\kernel{K}$ and $\init$ be a Markov transition kernel and a probability measure on $(\Xsp, \Xfd)$, respectively, and $\{ \pot[z] : z \in \Zsp \}$ a family of real-valued, positive, and measurable \emph{potential functions} on $(\Xsp, \Xfd)$. For  all vectors $\chunk{z}{k}{m} \in \Zsp^{m - k + 1}$, we define unnormalised transition kernels
$$
    \uk[\chunk{z}{k}{m}] : \Xsp \times \Xfd \ni (x_k, A) 
    \mapsto \idotsint \1_A(x_{m + 1}) \prod_{\ell = k}^m 
    \pot[z_\ell](x_\ell) \, \mk(x_\ell, \rmd x_{\ell + 1}),
$$
with the convention $\uk[\chunk{z}{k}{m}](x, A) = \delta_x(A)$ if $m < k$ (where $\delta_x$ denotes the Dirac mass located at $x$),  
and probability measures 
\begin{equation} \label{eq:def:pred}
    \pred[\chunk{z}{k}{m}] : \Xfd \ni A 
    \mapsto \frac{\init \uk[\chunk{z}{k}{m}] \1_A}
    {\init \uk[\chunk{z}{k}{m}] \1_\Xsp}. 
\end{equation}
Using these definitions we may, given a sequence $\{ z_n \}_{n \in \nset}$ of \emph{perturbations} in $\Zsp$, express the \emph{Feynman-Kac distribution flow} $\{ \pred[\chunk{z}{0}{n}] \}_{n \in \nset}$ recursively as 
\begin{equation} \label{eq:pred:rec}
    \pred[\chunk{z}{0}{n}] 
    = \frac{\pred[\chunk{z}{0}{n - 1}] \uk[z_n]}{\pred[\chunk{z}{0}{n - 1}]  \uk[z_n] \1_\Xsp}, \quad n \in \nset 
\end{equation}
(where, by the previous convention, $\pred[\chunk{z}{0}{- 1}] = \init$). Even though the previous model may be applied in a non-temporal context, we will often refer to the index $n$ as ``time''. 

\begin{example}[partially dominated state-space models] \label{example:state:space:model}
Let $(\Xsp, \Xfd)$ be a measurable space, $\hk : \Xsp \times \Xfd \rightarrow [0, 1]$ a Markov transition kernel, and $\init$ a probability measure on $(\Xsp, \Xfd)$ (the latter being referred to as the \emph{initial distribution}). In addition, let $(\Ysp, \Yfd)$ be another measurable space and $\ed : \Xsp \times \Ysp \rightarrow \rset_+$ a Markov transition density with respect to some reference measure $\refm$ on $(\Ysp, \Yfd)$. By a general state-space model we mean the canonical version of the bivariate Markov chain $\{ (X_n, Y_n) \}_{n \in \nset}$ having transition kernel 
$$
    \Xsp \times \Ysp \times \Xfd \varotimes \Yfd \ni ((x, y), A) \mapsto \iint \1_A(x', y') \ed(x', y') \, \refm(\rmd y') \, \hk(x, \rmd x')
$$
and initial distribution 
$$
    \Xfd \varotimes \Yfd \ni A \mapsto \iint \1_A(x, y) \ed(x, y) \, \refm(\rmd y) \, \init(\rmd x).  
$$ 
Here the marginal process $\{ X_n \}_{n \in \nset}$, referred to as the \emph{state process}, is only partially observed through the \emph{observation process} $\{ Y_n \}_{n \in \nset}$. For the model  $\{ (X_n, Y_n) \}_{n \in \nset}$ defined in this way, 
\begin{itemize}
    \item[(i)] the state process is a Markov chain with transition kernel $\mk$ and initial distribution $\init$, 
    \item[(ii)] the observations are, given the states, conditionally independent and such that the marginal conditional distribution of each $Y_n$ depends on $X_n$ only and has density $\pot(X_n, \cdot)$
\end{itemize}
(we refer to \cite[Section~2.2]{cappe:moulines:ryden:2005} for details). When operating on a well-specified state-space model, a key ingredient is typically the computation of the flow of \emph{predictor distributions}, where the predictor $\pred[\chunk{y}{0}{n - 1}]$ at time $n \in \nset$ is defined as the conditional distribution of the state $X_n$Ê given the record $\chunk{y}{0}{n - 1} \in \Ysp^n$ of realised historical observations up to time $n - 1$. Using Bayes' formula (see, e.g., \cite[Section~3.2.2]{cappe:moulines:ryden:2005} for details), it is straightforwardly shown that the predictor flow satisfies a perturbed Feynman-Kac recursion \eqref{eq:pred:rec} with $(\Xsp, \Xfd)$, $\hk$, and $\init$ given above, the observations $\{ Y_n \}_{n \in \nset}$ playing the role of perturbations (i.e., $\Zsp \gets \Ysp$ and $\Zfd \gets \Yfd$), and the local likelihood functions $\{Ê\ed(\cdot, y) : y \in \Ysp \}$ playing the role of potential functions $\{Ê\pot[y] : y \in \Ysp\}$. We will return to this framework in Section~\ref{sec:numerical:study}. 
\end{example}

\subsection{Sequential Monte Carlo methods}
SMC methods approximate online the Feynman-Kac flow generated by \eqref{eq:pred:rec} and a given sequence $\{ z_n \}_{n \in \nset}$ of perturbations by propagating recursively a random sample  $\{ \epart{n}{i} \}_{i = 1}^\N$ of $\Xsp$-valued \emph{particles}. More specifically, given a particle sample $\{ \epart{n}{i} \}_{i = 1}^\N$ \emph{targeting} $\pred[\chunk{z}{0}{n - 1}]$ in the sense that for all $h \in \bmf{\Xfd}$, $\predpart[\chunk{z}{0}{n - 1}] h \backsimeq \pred[\chunk{z}{0}{n - 1}] h$ as $\N$ tends to infinity, where  
$$
    \predpart[\chunk{z}{0}{n - 1}]: \Xfd \ni A \mapsto \frac{1}{\N} \sum_{i = 1}^\N \1_A(\epart{n}{i})
$$
denotes the empirical measure associated with the particles, an updated particle sample $\{ \epart{n + 1}{i} \}_{i = 1}^\N$ approximating $\pred[\chunk{z}{0}{n}]$ is, as the perturbation $z_n$ becomes accessible, formed by Algorithm~\ref{alg:SMC}. 

\bigskip
\begin{algorithm}[H] \label{alg:SMC}
    \KwData{$\{ \epart{n}{i} \}_{i = 1}^\N$, $z_n$}
    \KwResult{$\{ \epart{n + 1}{i} \}_{i = 1}^\N$}
    set $\wgtsum{n} \gets 0$\;
    \For{$i = 1 \to \N$}{
        set $\wgt{n}{i} \gets \pot[z_n](\epart{n}{i})$\;
        set $\wgtsum{n} \gets \wgtsum{n} + \wgt{n}{i}$\;
    }
    \For {$i = 1 \to \N$}{
        draw $\ind{n + 1}{i} \sim \cat(\{ \wgt{n}{\ell} / \wgtsum{n} \}_{\ell = 1}^N)$\;
        draw $\epart{n + 1}{i} \sim \mk(\epart{n}{\ind{n + 1}{i}}, \cdot)$\;
    }
    \caption{SMC particle update}
\end{algorithm}
\bigskip
(In the algorithm above, $\cat( \{ \wgt{n}{\ell} / \wgtsum{n} \}_{\ell = 1}^N)$ denotes the categorical distribution induced by the normalised particle weights $\{ \wgt{n}{\ell} / \wgtsum{n} \}_{\ell = 1}^N$.) Algorithm~\ref{alg:SMC} is initialised at time $n = 0$ by drawing $\{ \epart{0}{i} \}_{i = 1}^\N \sim \init^{\varotimes \N}$. For all $n \in \nset$ and all $h \in \bmf{\Xfd}$, the convergence, as $\N$ tends to infinity, of $\predpart[\chunk{z}{0}{n - 1}] h$ to $\pred[\chunk{z}{0}{n - 1}] h$ Êcan be established in several probabilistic senses. In particular, the first CLT for SMC methods was provided by \cite{delmoral:guionnet:1999}, establishing that 
\begin{equation} \label{eq:CLT}
\sqrt{\N} \left( \predpart[\chunk{z}{0}{n - 1}] h - \pred[\chunk{z}{0}{n - 1}] h \right) \dlim \variance{\chunk{z}{0}{n - 1}}(h) Z,  
\end{equation}
where $Z$ is standard normally distributed and the asymptotic variance is given by $\variance{\chunk{z}{0}{n - 1}} \eqdef \variance{\chunk{z}{0}{n - 1}}[0]$ with 
\begin{equation} \label{eq:def:as:var}
\variance[2]{\chunk{z}{0}{n - 1}}[\ell] : \bmf{\Xfd} \ni h \mapsto \sum_{m = \ell}^n \frac{\pred[\chunk{z}{0}{m - 1}] \{ \uk[\chunk{z}{m}{n - 1}](h - \pred[\chunk{z}{0}{n - 1}] h) \}^2 }{(\pred[\chunk{z}{0}{m - 1}] \uk[\chunk{z}{m}{n - 1}] \1_{\Xsp})^2}
\end{equation}
(see also \cite{chopin:2004,kuensch:2005,douc:moulines:2008} for similar results). The fact that we in \eqref{eq:def:as:var} define a truncated version of the variance with only $n - \ell + 1$ terms will be clear later on. In the coming section we propose a lag-based, numerically stable estimator of the sequence $\{ \variance[2]{\chunk{z}{0}{n - 1}} \}_{n \in \nset}$ of asymptotic variances. The estimator approximates $\{Ê\variance[2]{\chunk{z}{0}{n - 1}} \}_{n \in \nset}$ online, as $n$ increases, under constant computational complexity and memory requirements. Importantly, the estimator is obtained as a by-product of the particle filter output and does not require additional simulations. The numerical stability is obtained at the price of a small bias, which may be controlled under weak assumptions on the mixing properties of the model.

\section{A lag-based variance estimator}
\label{sec:estimator}
\subsection{The variance estimator proposed in \cite{chan:lai:2013}}
\label{sec:the:Lai:estimator}
Since Algorithm~\ref{alg:SMC} resamples the particles at each time step, the particle cloud may be associated with a tree describing the genealogical lineages of the particles. The estimators proposed in \cite{chan:lai:2013} and \cite{lee:whiteley:2016} are based on the particles' \emph{Eve indices} $\{ \eve{n}{i} \}_{i = 1}^\N$ (the terminology is adopted from \cite{lee:whiteley:2016}), which are, for all $n \in \nset$, defined as the indices of the time-zero ancestors of the particles $\{ \epart{n}{i} \}_{i = 1}^\N$. More specifically, the Eve indices may, for all $i \in \intvect{1}{\N}$, be computed recursively in Algorithm~\ref{alg:SMC} (just after Line~6) by letting 
$$
    \eve{n}{i} \eqdef
    \begin{cases}
        i & \mbox{for } n = 0, \\
        \eve{n - 1}{\ind{n}{i}} & \mbox{for } n \in \nsetpos. 
    \end{cases}
$$
Using the Eve indices, H.~P. Chan and T.~L.~Lai proposed, in \cite{chan:lai:2013}, for all $n \in \nset$, $\varest[2]{\chunk{z}{0}{n - 1}}(h)$, with
\begin{equation} \label{eq:Lai:estimator}
    \varest[2]{\chunk{z}{0}{n - 1}} : \bmf{\Xfd} \ni h \mapsto \frac{1}{\N} \sum_{i = 1}^\N \left( \sum_{j : \eve{n}{j} = i} \left\{ h(\epart{n}{j}) - \predpart[\chunk{z}{0}{n - 1}] h \right\} \right)^2,
\end{equation}
as an estimator of $\variance[2]{\chunk{z}{0}{n - 1}}(h)$ for all $h \in \bmf{\Xfd}$. As mentioned in the introduction, we will refer to this estimator as the CLE. (More precisely, in \cite{chan:lai:2013}, focus was set on the \emph{updated} distribution flows discussed in Section~\ref{sec:updated:measures} below; the adaptation is however straightforward.) In \cite{lee:whiteley:2016}, a generalisation of the CLE, allowing the particle population size $\N$ to vary between SMC iterations, is presented. As the main result of \cite{chan:lai:2013}, the consistency, as $\N$ tends to infinity, of the CLE is established; see also \cite[Theorem~1 and Corollary~1]{lee:whiteley:2016} for a generalisation. 

The CLE is indeed remarkable, as it allows the variance to be estimated online in a single run of the particle filter with no further simulation. Nevertheless, as explained in the introduction, the previous estimator has a serious flaw which is related to the well-known \emph{particle path depletion phenomenon} of SMC algorithms. More specifically, resampling the particles systematically at each time  leads without exception to a random time point before which all the genealogical traces coincide; we refer again to \cite{jacob:murray:rubenthaler:2015}, which provides a time uniform $\ordo(\N \log \N)$ bound on the expected number of generations back in time to this most recent common ancestor. Thus, as $n$ increases, the sets $\{ j \in \intvect{1}{\N} : \eve{n}{j} = i \}$ will eventually be empty for all indices $i \in \intvect{1}{\N}$ except one, say, $i_0$, for which $\{ j \in \intvect{1}{\N} : \eve{n}{j} = i_0 \} = \intvect{1}{\N}$. As a consequence, eventually, $\varest[2]{\chunk{z}{0}{n - 1}}(h) = 0$ for all $h \in \bmf{\Xfd}$, which makes the estimator impractical. In the next section, we propose a simple modification of the CLE that stabilises numerically the same at the cost of a negligible, controllable bias. 

\subsection{Our estimator}
\label{sec:our:estimator}

The estimator that we propose is based on the simple idea of stabilising numerically the CLE by tracing, backwards in time, only a few generations of the particle genealogy, rather than tracing the history all the way back to the time-zero ancestors. In our approach, the Eve indices will be replaced by
\emph{Enoch indices}\footnote{Two figures named Enoch appear in the 2nd as well as the 6th generations of the Genealogies of Genesis, as the son of Cain and the son-son-son-son-son of Seth, respectively.} defined, for all $i \in \intvect{1}{\N}$ and $m \in \nset$, recursively as  
\begin{equation} \label{eq:def:Enoch}
\enoch{m}{n}{i} \eqdef 
\begin{cases}
i & \mbox{for } n = m, \\
\enoch{m}{n - 1}{\ind{n}{i}}  & \mbox{for } n > m.  
\end{cases}
\end{equation}
In other words, for all $n \in \nset$, $m \in \intvect{1}{n}$, and $i \in \intvect{1}{\N}$, $\epart{m}{\enoch{m}{n}{i}}$ is the ancestor of $\epart{n}{i}$ at time $m$. Now, let $\lag \in \nset$ be some fixed number, referred to as the lag, and define $\lagtime{n}{\lag} \eqdef (n - \lambda) \vee 0$; then, we propose $\varest[2]{\chunk{z}{0}{n - 1}}[\lambda](h)$, with  
\begin{equation} \label{eq:estimator}
\varest[2]{\chunk{z}{0}{n - 1}}[\lambda] : \bmf{\Xfd} \ni h \mapsto \frac{1}{\N} \sum_{i = 1}^\N \left( \sum_{j : \enoch{\lagtime{n}{\lambda}}{n}{j} = i} \left\{ h(\epart{n}{j}) - \predpart[\chunk{z}{0}{n - 1}] h \right\} \right)^2,
\end{equation}
as an estimator of the variance $\variance[2]{\chunk{z}{0}{n - 1}}(h)$ for all $n \in \nset$, $\chunk{z}{0}{n - 1} \in \Zsp^n$, and $h \in \bmf{\Xfd}$. 
Online computation of the Enoch indices $\{ \enoch{\lagtime{n}{\lambda}}{n}{i} \}_{i =Ê1}^\N$ requires the propagation of a window $\{Ê\enoch{\lagtime{n}{\lag}}{n}{i}, \ldots,  \enoch{n}{n}{i} \}_{i = 1}^\N$ of indices; see Algorithm~\ref{alg:fixed-lag:SMC} for a pseudo-code.  As the length of the window is bounded by $\lag + 1$, the memory demand of the estimator is $\ordo(\lag \N)$ independently of $n$. Moreover, since genealogical tracing has a linear complexity in $\N$, the total complexity of the estimator is $\ordo(\lag \N)$, again independently of $n$. 

\bigskip
\begin{algorithm}[H] \label{alg:fixed-lag:SMC}
    \KwData{$\{ \epart{n}{i} \}_{i = 1}^\N$, $\{Ê\enoch{\lagtime{n}{\lag}}{n}{i}, \ldots,  \enoch{n}{n}{i} \}_{i = 1}^\N$, $z_n$}
    \KwResult{$\{ \epart{n + 1}{i} \}_{i = 1}^\N$, $\{Ê\enoch{\lagtime{(n + 1)}{\lag}}{n + 1}{i}, \ldots,  \enoch{n + 1}{n + 1}{i} \}_{i = 1}^\N$}
    set $\wgtsum{n} \gets 0$\;
    \For{$i = 1 \to \N$}{
        set $\wgt{n}{i} \gets \pot[z_n](\epart{n}{i})$\;
        set $\wgtsum{n} \gets \wgtsum{n} + \wgt{n}{i}$\;
    }
    \For {$i = 1 \to \N$}{
        draw $\ind{n + 1}{i} \sim \cat(\{ \wgt{n}{\ell} / \wgtsum{n} \}_{\ell = 1}^N)$\;
        draw $\epart{n + 1}{i} \sim \mk(\epart{n}{\ind{n + 1}{i}}, \cdot)$\;
        \For{$m = \lagtime{(n + 1)}{\lambda} \to n$}{
            set $\enoch{m}{n + 1}{i} \gets \enoch{m}{n}{\ind{n + 1}{i}}$\;
        }
        set $\enoch{n + 1}{n + 1}{i} \gets i$\;
    }
    \caption{SMC particle and Enoch-index update}
\end{algorithm}
\bigskip

For $n = 0$, Algorithm~\ref{alg:fixed-lag:SMC} is initialised by drawing $\{ \epart{0}{i} \}_{i = 1}^\N \sim \init^{\varotimes \N}$ and setting $\enoch{0}{0}{i} \gets i$ for all $i \in \intvect{1}{\N}$. At the end of the algorithm, after the second \textbf{for}-loop, an estimate 
$$
    \varest[2]{\chunk{z}{0}{n}}[\lambda](h) = \frac{1}{\N} \sum_{i = 1}^\N \left( \sum_{j : \enoch{\lagtime{(n + 1)}{\lambda}}{n + 1}{j} = i} \{ h(\epart{n + 1}{j}) - \predpart[\chunk{z}{0}{n}] h \} \right)^2
$$
of $\variance[2]{\chunk{z}{0}{n}}[\lambda](h)$ may be formed for all $h \in \bmf{\Xfd}$. 

\subsection{Variance estimators for flows of updated distributions}
\label{sec:updated:measures}

Some applications involve approximation of the \emph{updated} measures 
\begin{equation} \label{eq:def:filt}
    \filt[\chunk{z}{k}{m}] : \Xfd \ni A 
    \mapsto \frac{\init \uk[\chunk{z}{k}{m - 1}] (\pot[z_m] \1_A)}
    {\init \uk[\chunk{z}{k}{m - 1}] (\pot[z_m] \1_\Xsp)},
\end{equation}
for $\chunk{z}{k}{m} \in \Zsp^{m - k + 1}$, rather than the measures defined by \eqref{eq:def:pred}.   
\begin{example}[partially dominated state-space models, revisited]
In the case of the partially dominated state-space models discussed in Example~\ref{example:state:space:model}, the updated measures $\{ \filt[\chunk{y}{0}{n}] \}_{n \in \nset}$ defined through \eqref{eq:def:filt} are the \emph{filter distributions}; more precisely, in this context, for all $n \in \nset$, $\filt[\chunk{y}{0}{n}]$ is the conditional distribution of the state $X_n$ given the realised observations $\chunk{y}{0}{n} \in \Ysp^{n + 1}$ up to time $n$ (i.e., \emph{including} the last observation $y_n$). 
\end{example}
Since for all $h \in \bmf{\Xfd}$, by normalisation, 
$$
    \filt[\chunk{z}{k}{m}] h = \frac{\pred[\chunk{z}{k}{m - 1}](\pot[z_m] h)}{\pred[\chunk{z}{k}{m - 1}] \pot[z_m]},
$$
the flow $\{ \filt[\chunk{z}{0}{n}] \}_{n \in \nset}$ of updated distributions is naturally approximated by the flow of weighted empirical measures 
\begin{equation} \label{eq:def:particle:filter}
    \filtpart[\chunk{z}{0}{n}] : A \ni \Xfd \mapsto \frac{\predpart[\chunk{z}{k}{m - 1}](\pot[z_m] \1_A)}{\predpart[\chunk{z}{k}{m - 1}] \pot[z_m]} = \sum_{i = 1}^\N \frac{\wgt{n}{i}}{\wgtsum{n}} \1_A(\epart{n}{i}),
\end{equation}
for some given sequence $\{ z_n \}_{n \in \nset}$ of perturbations, where the weights $\{ \wgt{n}{i} \}_{i = 1}^\N$ and the weight sum $\wgtsum{n}$ are computed in Algorithm~\ref{alg:SMC}. By the normality \eqref{eq:CLT} and the consistency \eqref{eq:particle:filter:consistency} one obtains, using Slutsky's theorem, for all $\chunk{z}{0}{n} \in \Zsp^{n + 1}$, the central limit theorem 
\begin{equation} \label{eq:CLT:updated:measures}
    \sqrt{\N} \left( \filt[\chunk{z}{0}{n}] h - \filt[\chunk{z}{0}{n}] h \right) \dlim \filtvariance{\chunk{z}{0}{n}}(h) Z,  
\end{equation}
as $\N$ tends to infinity, where $Z$ is standard normally distributed and the asymptotic variance is given by $\filtvariance[2]{\chunk{z}{0}{n}}(h) = \filtvariance[2]{\chunk{z}{0}{n}}[0](h)$ with 
\begin{equation} \label{eq:def:as:var:updated:measures}
    \filtvariance[2]{\chunk{z}{0}{n}}[\ell] : \bmf{\Xfd} \ni h \mapsto \frac{\variance[2]{\chunk{z}{0}{n - 1}}[\ell](\pot[z_n] \{ h - \filt[\chunk{z}{0}{n}] h \})}{(\pred[\chunk{z}{0}{n - 1}] \pot[z_n])^2}
\end{equation}
(where $\variance{\chunk{z}{0}{n - 1}}[\ell]$ is defined in \eqref{eq:def:as:var} for the original Feynman-Kac particle model). In the case $\ell = 0$, the expression \eqref{eq:def:as:var:updated:measures} is found also in \cite[Eqn.~(17)]{douc:moulines:olsson:2014}. In the light of \eqref{eq:def:as:var:updated:measures}, casting our fixed-lag approach into the framework of updated Feynman-Kac models yields the estimator 
\begin{multline} \label{eq:def:var:est:updated:measures}
    \varestfilt[2]{\chunk{z}{0}{n}}[\lambda] : \bmf{\Xfd} \ni h \mapsto \frac{\varest[2]{\chunk{z}{0}{n - 1}}[\lambda](\pot[z_n] \{h - \filtpart[\chunk{z}{0}{n}] h\})}{(\predpart[\chunk{z}{0}{n - 1}] \pot[z_n])^2} \\ 
    = \N \sum_{i = 1}^\N \left( \sum_{j : \enoch{\lagtime{n}{\lambda}}{n}{j} = i} \frac{\wgt{n}{i}}{\wgtsum{n}} \left\{ h(\epart{n}{j}) - \filtpart[\chunk{z}{0}{n}] h \right\} \right)^2
\end{multline}
for some suitable lag $\lag \in \nset$ (where the equality stems from the fact that $\predpart[\chunk{z}{0}{n - 1}] \{Ê\pot[z_n](h -  \filtpart[\chunk{z}{0}{n}] h) \} = 0$).

\section{Theoretical results}
\label{sec:theoretical:results}
\subsection{Main assumptions}

In the following we assess the theoretical properties of our estimator. All proofs are presented in Appendix~\ref{sec:proofs}. We preface our main assumptions by the definition of an \emph{$r$-local Doeblin set}.
\begin{definition} \label{defi:local-Doeblin}
Let $r \in \nsetpos$. A set $C \in \Xfd$ is $r$-\emph{local Doeblin} with respect to $\mk$Ê and $\pot$ if there exist positive functions $\varepsilon^-_C: \Zsp^r \to \rset^+$ and $\varepsilon^+_C: \Zsp^r \to \rset^+$, a family $\{ \probdoeblin{C}{\bar{z}} ; \bar{z} \in \Zsp^r \}$ of probability measures, and a family $\{\varphi_C \langle \bar{z} \rangle ; \bar{z} \in \Zsp^r \}$ of positive functions such that for all $\bar{z} \in \Zsp^r$, $\probdoeblin{C}{\bar{z}}(C) =1$ and for all $A \in \Xfd$ and $x \in C$,
\begin{equation*} \label{eq:definition-LD-set}
	\varepsilon^-_C \langle \bar{z} \rangle \, \varphi_C \langle \bar{z} \rangle(x) \, \probdoeblin{C}{\bar{z}} (A) \leq \uk[\bar{z}](x, A \cap C) \leq \varepsilon^+_C \langle \bar{z} \rangle \, \varphi_C \langle \bar{z} \rangle(x) \, \probdoeblin{C}{\bar{z}}(A).
\end{equation*}
\end{definition}

Our theoretical analysis is driven by the following assumptions. 

\begin{hypA} \label{assum:likelihoodDrift}
The pertubation process $\{ \per_n \}_{n \in \zset}$, taking on values in $(\Zsp, \Zfd)$, is strictly stationary and ergodic. Moreover, there exist an integer $r \in \nsetpos $ and a set $K \in \Zfd^{\varotimes r}$ such that the following holds.
\begin{enumerate}[(i)]
    \item \label{item:condition-L-K}
    The process $\{Ê\perblock_n \}_{n \in \zset}$, where $\perblock_n \eqdef \chunk{\per}{n r}{(n + 1)r - 1}$, is ergodic and such that $\prob(\bar{\per}_0 \in K) > 2/3$.
    \item For all $\eta > 0$ there exists an $r$-local Doeblin set $C \in \Xfd$ such that for all $\chunk{z}{0}{r - 1} \in K$,
    \begin{equation*}
        \label{eq:bound-eta-G}
        \sup_{x \in C^\complement} \uk[\chunk{z}{0}{r-1}] \1_\Xsp(x) \leq \eta \sup_{x \in \Xsp} \uk[\chunk{z}{0}{r - 1}] \1_\Xsp(x) < \infty
    \end{equation*}
    and
    \begin{equation*} \label{eq:lower-bound}
        \inf_{\chunk{z}{0}{r - 1} \in K} \frac{\varepsilon_C^- \langle \chunk{z}{0}{r - 1} \rangle}{\varepsilon_C^+ \langle \chunk{z}{0}{r - 1} \rangle} > 0,
    \end{equation*}
    where the functions $\varepsilon^+_C$ and $\varepsilon^-_C$ are given in Definition~\ref{defi:local-Doeblin}.
    \item \label{item:condition-minoration}
    There exists a set $D \in \Xfd$ such that
    \begin{equation*}
        \label{eq:condition-minoration}
        \E \left[\ln^- \inf_{x \in D} \uk[\chunk{\per}{0}{r-1}] \1_D(x) \right] < \infty.
    \end{equation*}
\end{enumerate}
\end{hypA}

\begin{hypA}\label{assum:majo-g}
\begin{enumerate}[(i)]
    \item \label{item:assum-pos-g} For all $(x, z) \in \Xsp \times \Zsp$, $\pot[z](x) > 0$.
    \item \label{item:bdd:g} For all $z \in \Zsp$, $\| \pot[z] \|_\infty < \infty$.
    \item \label{item:assum-supg} $\E \left[ \ln^+ \| \pot[\per_0] \|_\infty \right] < \infty$.
\end{enumerate}
\end{hypA}

\begin{remark} \label{rem:simplified:assumptions}
    In the case $r = 1$, \A{assum:likelihoodDrift} may be replaced by the simpler assumption
    that there exists a set $K \in \Zfd$ such that the following holds.
    \begin{enumerate}[(i)]
        \item \label{item:condition-L-K-simplified}
        $\prob \left( \per_0 \in K \right) > 2 / 3$.
        \item \label{item:condition:local:doeblin:simplified}
        For all $\eta > 0$ there exists a $1$-local Doeblin set $C \in \Xfd$ such that for all $z \in K$,
        \begin{equation}
            \label{eq:bound-eta-G-simplified}
            \sup_{x \in C^\complement} \pot[z](x) \leq \eta \| \pot[z] \|_\infty < \infty.
        \end{equation}
        \item \label{item:mino-g-simple} There exists a set $D \in \Xfd$ satisfying
        $$
            \inf_{x \in D} \mk(x, D) > 0 \quad \mbox{and} \quad \E \left[ \ln^- \inf_{x \in D} \pot[\per_0](x) \right] < \infty.
        $$
    \end{enumerate}
    This simplified version will be used in Section~\ref{sec:numerical:study}.
\end{remark}

\subsection{Theoretical properties of the fixed-lag estimator}

As established by the following proposition, for all $n \in \nset$, $\chunk{z}{0}{n - 1} \in \Zsp^n$, Êand $\lag \in \nset$, the estimator $\varest[2]{\chunk{z}{0}{n - 1}}[\lambda]$ is a consistent for $\variance[2]{\chunk{z}{0}{n - 1}}[\lagtime{n}{\lag}]$, where the latter is given by \eqref{eq:def:as:var} with $\ell = \lagtime{n}{\lag}$. 

\begin{proposition} \label{prop:consistency:fixed:lag}
Assume \A{assum:majo-g}(\ref{item:assum-pos-g}--\ref{item:bdd:g}). Then for all $n \in \nset$, $\chunk{z}{0}{n - 1} \in \Zsp^n$, $\lag \in \nset$, and $h \in \bmf{\Xfd}$, as $\N \rightarrow \infty$, 
$$
\varest[2]{\chunk{z}{0}{n - 1}}[\lambda](h) \plim \variance[2]{\chunk{z}{0}{n - 1}}[\lagtime{n}{\lag}](h). 
$$
\end{proposition}

Now, define, for all $n \in \nset$ and $\chunk{z}{0}{n - 1} \in \Zsp^n$, the \emph{asymptotic bias}  
\begin{equation} \label{eq:def:bias}
\bias{\chunk{z}{0}{n - 1}}{\lag} : \bmf{\Xfd} \ni h \mapsto \variance[2]{\chunk{z}{0}{n - 1}}(h) - \variance[2]{\chunk{z}{0}{n - 1}}[\lagtime{n}{\lag}](h), 
\end{equation}
which is always nonnegative. For the integer $r \in \nsetpos$ and the set $D \in \Xfd$ given in \A{assum:likelihoodDrift}, introduce the class 
\begin{equation} \label{eq:class:initial:distributions}
\mdr(D, r) \eqdef \left\{ \init \in \probmeas{\Xfd} : \E \left[ \ln^- \init \uk[\chunk{\per}{0}{\ell - 1}] \1_D \right] < \infty \mbox{ for all $\ell \in \intvect{0}{r}$} \right \}
\end{equation}
of initial distributions. Then the following theorem, which is the main result of this paper, states that the asymptotic bias can be controlled uniformly in time by the lag $\lambda$. In particular, the bias decreases with $\lambda$ at a geometric rate. 

\begin{theorem} \label{thm:tightness:bias}
Assume \A[assum:likelihoodDrift]{assum:majo-g}. Then there exist a constant $\rate \in (0, 1)$ and a sequence $(\bd_k)_{k \in \nsetpos}$ in $\rset_+$ such that for all $\init \in \mdr(D, r)$ (with $D$ and $r$ given by \A{assum:likelihoodDrift}), $\lag \in \nset$, $h \in \bmf{\Xfd}$, and $k \in \nset$,
\begin{equation} \label{eq:bias:bound}
\sup_{n \in \nset} \prob \left( \frac{\bias{\chunk{\per}{0}{n - 1}}{\lag}(h)}{\rate^{\lambda + 1} \| h \|_{\infty}^2} > c_k  \right) \leq \frac{1}{k}.   
\end{equation}
\end{theorem}

\begin{remark}
We stress that the sequence $\{ \bd_k \}_{k \in \nsetpos}$ as well as the mixing rate $\rho$ in the previous theorem depend exclusively on the properties of the model and not on the lag size $\lag$, the objective function $h$, or the initial distribution $\init$ (as long as this belongs to $\mdr(D, r)$). This means that at any level specified by $k$, the asymptotic bias can be controlled uniformly in time at a geometric rate in $\lag$.
\end{remark} 

The strength of these bounds will be illustrated numerically in Section~\ref{sec:numerical:study}. 

\subsection{Extensions to variance estimators for flows of updated distributions}

In the following we confirm that the previous results can be extended to flows of updated Feynman-Kac distributions (recall the definitions in Section~\ref{sec:updated:measures}). 

\begin{proposition} \label{prop:consistency:fixed:lag:updated:case}
Assume \A{assum:majo-g}(\ref{item:assum-pos-g}--\ref{item:bdd:g}). Then for all $n \in \nset$, $\chunk{z}{0}{n} \in \Zsp^{n + 1}$, $\lag \in \nset$, and $h \in \bmf{\Xfd}$, as $\N \rightarrow \infty$, 
$$
\varestfilt[2]{\chunk{z}{0}{n}}[\lambda](h) \plim \filtvariance[2]{\chunk{z}{0}{n}}[\lagtime{n}{\lag}](h). 
$$
\end{proposition}

Now, as in the previous section, define, for all $n \in \nset$ and $\chunk{z}{0}{n} \in \Zsp^{n + 1}$, the asymptotic bias   
\begin{equation} \label{eq:def:bias:filter}
\biasfilt{\chunk{z}{0}{n}}{\lag} : \bmf{\Xfd} \ni h \mapsto \filtvariance[2]{\chunk{z}{0}{n}}(h) - \filtvariance[2]{\chunk{z}{0}{n}}[\lagtime{n}{\lag}](h). 
\end{equation} 

\begin{theorem} \label{thm:tightness:bias:filter}
Assume \A[assum:likelihoodDrift]{assum:majo-g}. Then the statement of Theorem~\ref{thm:tightness:bias} holds true when $\bias{\chunk{\per}{0}{n - 1}}{\lag}$ is replaced by $\biasfilt{\chunk{\per}{0}{n}}{\lag}$. 
\end{theorem}

\subsection{Bounds on the asymptotic bias under strong mixing assumptions}

Before continuing to the numerical part of the paper, we provide, for the sake of completeness, a stronger bound on the asymptotic bias under the following \emph{strong mixing} assumption, which is standard in the literature of SMC analysis (see \cite{delmoral:guionnet:2001} and, e.g., \cite{delmoral:2004,cappe:moulines:ryden:2005,crisan:heine:2008} for refinements) and points to applications where the state space $\Xsp$ is a compact set.   

\begin{hypB} \label{ass:strong:mixing}
    \begin{enumerate}[(i)]
    \item \label{ass:strong:mixing-1} There exist constants $0 < \mdlow < \mdup < \infty$ and a probability measure $\mu \in \probmeas{\Xfd}$ such that for all $x \in \Xsp$ and $A \in \Xfd$, 
    $$
        \mdlow \mu(A) \leq \mk(x, A) \leq \mdup \mu(A).
    $$
    \item \label{ass:strong:mixing-2} There exist constants $0 < \potlow < \potup < \infty$ such that for all $z \in \Zsp$, $\| \pot[z] \|_\infty \leq \potup$ and for all $(x, z) \in \Xsp \times \Zsp$, $\mk \pot[z](x) \geq \potlow$.
    \end{enumerate}
\end{hypB}

Under \B{ass:strong:mixing}(\ref{ass:strong:mixing-1}), define $\varrho \eqdef 1 - \mdlow / \mdup$; then, instead of bounding the asymptotic bias stochastically (as in Theorem~\ref{thm:tightness:bias}), the following theorem provides a \emph{deterministic} bound on the same. The proof is found in Section~\ref{sec:proof:strong:bias:bound}. 

\begin{theorem} \label{thm:strong:bias:bound}
Assume \B{ass:strong:mixing}. Then there exists a constant $c \in \rsetpos$ such that for all $n \in \nset$, $\chunk{z}{0}{n - 1} \in \Zsp^n$, $\lag \in \nset$, and $h \in \bmf{\Xfd}$, 
\begin{equation*} \label{eq:strong:bias:bound}
\bias{\chunk{z}{0}{n - 1}}{\lag}(h) \leq 
\begin{cases} 
0 & \mbox{for } n \leq \lag, \\
c \varrho^{2(\lambda + 1)} \| h \|_{\infty}^2 & \mbox{for } n > \lag. 
\end{cases} 
\end{equation*}
\end{theorem}

\section{Application to state-space models}
\label{sec:numerical:study}
\subsection{Predicting log-volatility}
\label{sec:stochastic:volatility}

As a first numerical illustration we consider the problem of computing the predictor flow in the---now classical---stochastic volatility model 
\begin{equation} \label{eq:sto:vol:model}
	\begin{split}
		X_{n + 1} &= \varphi X_n + \sigma U_{n + 1}, \\
		Y_n &= \beta \exp \left( X_n / 2 \right) V_n,   
	\end{split}
	\quad n \in \nset, 
\end{equation}
proposed in \cite{hull:white:1987}, where $\{ U_n \}_{n \in \nsetpos}$ and $\{ V_n \}_{n \in \nset}$ are sequences of uncorrelated standard Gaussian noise. In this model, where only the process $\{ Y_n \}_{n \in \nset}$ is observable, $\beta$ is a constant scaling factor, $\varphi$ is the \emph{persistence}, and $\sigma$ is the \emph{volatility of the log-volatility}. In the case where $|\varphi| < 1$, the state process $\{ X_n \}_{n \in \nset}$ is stationary with a Gaussian invariant distribution having zero mean and variance $\tilde{\sigma}^2 \eqdef \sigma^2 / (1 - \varphi^2)$.  

It is easily checked that the model $\{ (X_n, Y_n) \}_{n \in \nset}$ defined by \eqref{eq:sto:vol:model} is indeed a state-space model in the notion of Example~\ref{example:state:space:model}, and in this case $\Xsp = \Ysp = \rset$, $\Xfd = \Yfd = \mathcal{B}(\rset)$, and  
\[
	\begin{split}
		\mk(x, A) &= \int_A \phi(x' ; \varphi x, \sigma^2) \, \rmd x', \\
		\pot(x, y) &= \phi(y ; 0, \beta^2 \me^x),  
	\end{split}
\]
where $(x, y) \in \rset^2$, $A \in \mathcal{B}(\rset)$, and $\phi(\cdot ; \mu, \varsigma^2)$ denotes the density of the Gaussian distribution with expectation $\mu \in \rset$ and standard deviation $\varsigma > 0$. Consequently, the flow of one-step log-volatility predictors satisfies the perturbed Feynman-Kac recursion \eqref{eq:pred:rec} with $\pot[y] = \ed(\cdot, y)$, $y \in \Ysp$, and may thus be approximated using SMC methods. 

We check that the model satisfies the simplified assumptions in Remark~\ref{rem:simplified:assumptions} for the scenario where $|\varphi| < 1$, implying stationary state and observation processes, the latter with marginal stationary distribution 
\begin{equation} \label{eq:stationary:marginal:observations}
	\mathcal{B}(\rset) \ni A \mapsto \iint \1_A(y) \phi(y ; 0, \beta^2 \me^x) \phi(x ; 0, \tilde{\sigma}^2) \, \rmd x \, \rmd y. 
\end{equation}
For this purpose, we first note that for all $y \in \rset$,
$$
	\| \pot[y] \|_\infty = \frac{1}{|y| \sqrt{2 \pi}} \me^{-1/2}, 
$$
i.e., for all $(x, y) \in \rset^2$, 
\begin{equation*}
	\frac{\pot[y](x)}{\| \pot[y] \|_\infty} \propto |y| \exp \left( - \frac{x}{2} - \frac{y^2 \me^{- x}}{2 \beta^2} \right), 
\end{equation*}
where the right hand side tends to zero as $|x| \rightarrow \infty$ for all $y \in \rset$. We thus conclude that Conditions~\eqref{item:condition-L-K} and \eqref{item:condition:local:doeblin:simplified} in Remark~\ref{rem:simplified:assumptions} hold true for any compact set $K \subset \rset$ with probability exceeding $2/3$ under \eqref{eq:stationary:marginal:observations}. In this case, every compact set $C \subset \rset$ is $1$-local Doeblin with respect to Lebesgue measure. In addition, since 
\begin{equation} \label{eq:finite:second:moment:obs:process}
	\E \left[ÊY_0^2 \right] = \beta^2 \int \me^x \phi(x ; 0, \tilde{\sigma}^2) \, \rmd x < \infty, 
\end{equation}
Condition~\eqref{item:mino-g-simple} is satisfied for all compact sets $D \subset \rset$. Moreover, for such compact sets $D$, \eqref{eq:finite:second:moment:obs:process} implies that $\mdr(D, 1)$ contains all initial distributions $\init$ with $\init(D) > 0$; in particular,  $\mdr(D, 1)$ contains the invariant Gaussian distribution of the state process, which we use as initial distribution in the predictor Feynman-Kac recursion. 

In this context, we aim at computing the sequence $\{ \pred[\chunk{y}{0}{n - 1}] \operatorname{id} \}_{n \in \nset}$ of predictor means; however, due to the nonlinearity of the emission density $\pot$, closed-form expressions are beyond reach, and we thus apply SMC for this purpose. In the scenario considered by us, $(\beta, \varphi, \sigma) = (.641, .975, .165)$, which are the parameter estimates obtained in \cite[Example~11.1.1.2]{cappe:moulines:ryden:2005} on the basis of British pound/US dollar daily log-returns from 1~October 1981 to 28~June 1985. In this setting, the following two numerical experiments were conducted.  

\subsubsection{Lag size influence}
\label{sec:lag:size:influence}

First, in order to assess the dependence of the bias \eqref{eq:def:bias} on the lag $\lag$, a record comprising $600$ observations were generated by simulation of the model under the parameters above. We re-emphasise that even though we here, for simplicity, consider the scenario of a perfectly specified model, this is not required in our assumptions in Section~\ref{sec:theoretical:results} (as long as the observation sequence is stationary). By inputting, $100$ times, these observations into Algorithm~\ref{alg:fixed-lag:SMC} and running the same with $\N = 4000$ particles, $100$ replicates of the particle predictor mean at time $600$ were produced and furnished with estimates of the asymptotic variance at the same time point. For each replicate, the asymptotic variance was estimated using all the lags $\lag \in \{2, 10, 12, 14, 16, 18, 20, 22, 50, 100, 200, 600\}$, where the last one, $\lag = 600$, corresponds to the CLE. Moreover, the reference value $1.63$ of the asymptotic variance, again at the time point $600$, was obtained by the brute-force approach consisting in generating as many as $1000$ replicates of the particle predictor mean, again with $\N = 4000$ particles, and simply multiplying the sample variance of these replicates by $\N$. The outcome is reported numerically and visually in Table~\ref{tab:lags:means:stds:stovol} and Figure~\ref{fig:boxplots:stovol}, respectively, where the different boxes in Figure~\ref{fig:boxplots:stovol} correspond to different lags. For each box, the reference value is indicated by a blue-colored asterisk $(\ast)$ and the average estimate of each box is indicated by a red ditto. From Figure~\ref{fig:boxplots:stovol} it is evident that the clear bias introduced by the smallest lags ($\lag \in \{2, 10, 12\}$) decreases rapidly as $\lag$ increases, and for $\lag = 20$ the bias is more or less eliminated. After this, increasing further $\lag$ leads, as we may expect from the particle path degeneracy, to significant increase of variance and no further improvement of the bias. On the contrary, the fact that the cardinality of the set $\{Ê\enoch{\lagtime{n}{\lambda}}{n}{i} : i \in \intvect{1}{\N} \}$ decreases monotonously as $\lag$ increases, leading to constantly zero variance estimates for $n$ and $\lambda$ large enough, re-introduces bias also for large $\lag$. For the last replicate, the cardinalities of the sets $\{ \eve{n}{i} : i \in \intvect{1}{\N} \}$ and $\{Ê\enoch{\lagtime{n}{20}}{n}{i} : i \in \intvect{1}{\N} \}$ were $8$ and $140$, respectively.  
\begin{figure}[H] 
    \centering
    \includegraphics[width=0.8\textwidth]{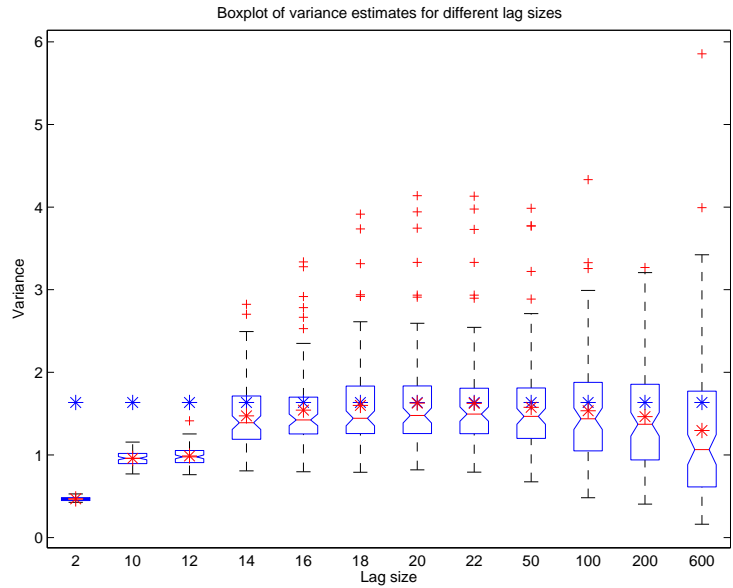} 
    \caption{Estimated asymptotic variances of the particle predictor mean at time $600$ in the stochastic volatility model \eqref{eq:sto:vol:model}. The particle population size is $\N = 4,\!000$. The boxes are based on $100$ replicates of Algorithm~\ref{alg:fixed-lag:SMC} and correspond to the different lags $\lag \in \{2, 10, 12, 14, 16, 18, 20, 22, 50, 100, 200, 600\}$. The reference value $1.63$, which is indicated by blue-colored asterisks $(\ast)$, was obtained as the sample variance (scaled by the number of particles) of $1000$ independent replicates of the particle predictor mean at time $600$ (again, Algorithm~\ref{alg:fixed-lag:SMC} used $\N = 4000$ particles). Red-colored asterisks indicate average estimates.} \label{fig:boxplots:stovol}
\end{figure}

\begin{table}[H] \label{tab:lags:means:stds}
    \begin{center}
        \begin{tabular}{c|c|c} \toprule
            $\lambda$ & Mean & St. dev. \\ \midrule 
            $2$     & $.47$   & $.02$ \\
            $10$   & $.96$   & $.09$ \\ 
            $12$   & $.99$   & $.11$ \\
            $14$   & $1.47$Ê& $.40$ \\
            $16$   & $1.54$ & $.49$ \\
            $18$   & $1.60$ & $.57$ \\
            $20$   & $1.63$ & $.62$ \\
            $22$   & $1.62$ & $.61$ \\
            $50$   & $1.58$ & $.61$ \\
            $100$ & $1.53$ & $.71$ \\
            $200$ & $1.46$ & $.71$ \\
            $600$ & $1.30$ & $.96$ \\
        \bottomrule
        \end{tabular} 
    \end{center}
    \caption{Means and standard deviations of the variance estimates reported in Figure~\ref{fig:boxplots:stovol}.} 
    \label{tab:lags:means:stds:stovol}
\end{table}

\subsubsection{Long-term stability}
\label{sec:long:term:stability}

In order to investigate numerically the long-term stability of our fixed-lag estimator, we executed Algorithm~\ref{alg:fixed-lag:SMC} on a considerably longer observation record comprising $3500$ values generated by simulation. The number of particles was set to $\N = 5000$.Ê Guided by the outcome of the previous experiment, we furnished the estimated predictor means with variance estimates obtained in the same sweep of the data using the fixed-lag estimator with $\lag = 20$. In parallel, the CLE was computed on the same realisation of the particle cloud. Finally, the brute-force approach estimating the asymptotic variances on the basis of $1200$ replicates of the predictor mean sequence was re-applied as a reference. 

Figure~\ref{fig:variance:evolution:short} displays the time evolution of the variance estimates over the initial $200$ time steps, where only estimates corresponding to every second time step have been plotted for readability. As evident from the top panel, the CLE targets well the reference values at most time points for this relatively limited time horizon, even though slight numerical instability may be discerned towards the very end of the plot. In addition, the fixed-lag estimator closes nicely the reference values for most time points with a variance that is somewhat smaller than that of the CLE. In order to display the estimators' long-term stability properties, Figure~\ref{fig:variance:evolution:long} provides the analogous plot for the full observation record comprising $3500$ values. Again, for readability, only estimates corresponding to every $35^\mathrm{th}$ time step have been plotted. Now, as clear from the top panel, the estimates produced by the CLE degenerate rapidly and after, say, $1500$ time steps the CLE loses track completely of the reference values. From time step $2871$, all particles in the cloud share the same Eve index, and the CLE collapses to zero. On the the contrary, from the bottom panel it is evident that the estimates delivered by the fixed-lag estimator stay numerically stable and closes well the reference values at most time points. In particular, the variance peak arising as a result of extreme state process behavior at time $3395$ is captured strikingly well by the fixed-lag estimator. 

\begin{figure}[H] 
\centering
\includegraphics[width=0.8\textwidth]{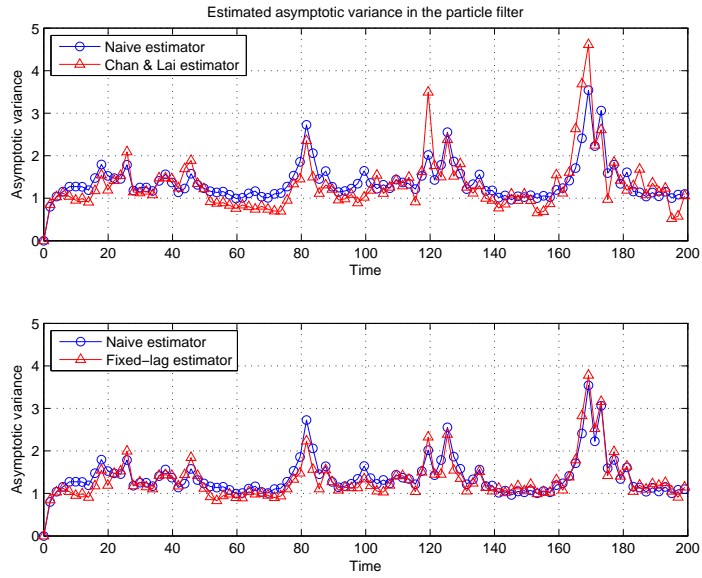}
\caption{Long-term evolution of estimated asymptotic variances of particle predictor means in the stochastic volatility model \eqref{eq:sto:vol:model}. For clarity, only every second estimate is plotted. 
The top panel displays variance estimates ($\circ$) produced using the CLE estimator along with reference values obtained as the sample variances (scaled by the number of particles) computed from $1200$ independent replicates of the particle predictor mean sequence. The bottom panel displays the analogous plot for estimates ($\circ$) produced using the fixed-lag estimator with $\lag = 20$. The number of particles was set to $\N = 5000$ is all cases.}
 \label{fig:variance:evolution:short}
\end{figure}

\begin{figure}[H] 
\centering
\includegraphics[width=0.8\textwidth]{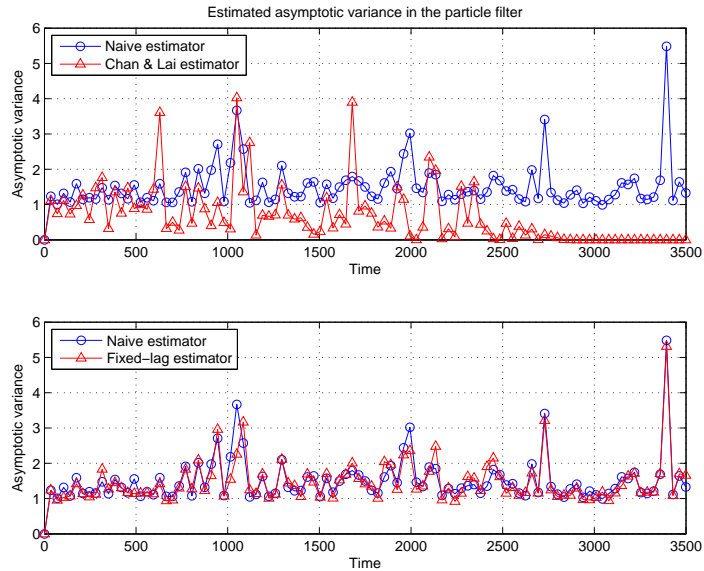}
\caption{The same plot as in Figure~\ref{fig:variance:evolution:short}, but now for the full observation record comprising $3500$ observations. For clarity, only every $35^{\mathrm{th}}$ estimate is plotted.}
\label{fig:variance:evolution:long}
\end{figure}

\subsection{SMC confidence bounds}
\label{sec:lin:Gauss:model}

The \emph{linear Gaussian state-space model}
\begin{equation} \label{eq:lin:Gauss:model}
\begin{split}
X_{n + 1} &= \varphi X_n + \sigma_u U_{n + 1}, \\
Y_n &= X_n + \sigma_v V_n,   
\end{split}
\quad n \in \nset, 
\end{equation}
where $\{ U_n \}_{n \in \nsetpos}$ and $\{ V_n \}_{n \in \nset}$ are again sequences of uncorrelated standard Gaussian noise, allows predictor means to be computed in a closed form using \emph{Kalman prediction} (see, e.g., \cite[Algorithm~5.2.9]{cappe:moulines:ryden:2005}). Thus, allowing comparisons with true quantities to be made, the class of linear Gaussian models is often used as a testing lab for SMC algorithms. In the setting where $|\varphi| < 1$, the state and observation processes are stationary, with zero mean Gaussian marginal stationary distributions with variances $\tilde{\sigma}^2$ and $\tilde{\sigma}^2 + \sigma_v^2$, respectively, where $\tilde{\sigma}^2 \eqdef \sigma_u^2 / (1 - \varphi^2)$. We let the former initialise the state process. 

Arguing along the lines of Section~\ref{sec:stochastic:volatility}, Assumptions~\A[assum:likelihoodDrift]{assum:majo-g} are checked straightforwardly also for this model. We leave the details to the reader.   

After having generated, for the parameterisation $(\varphi, \sigma_u, \sigma_v) = (.98, .2, 1)$, an observation record of $600$ observations, the experiment in Section~\ref{sec:lag:size:influence} was repeated for the same set of lag sizes $\lag$. As in the stochastic volatility example, the evolution of the $\N = 4000$ particles followed the same model dynamics as that generating the observations, and the reference value $1.102$ of the variance at $n = 600$ was obtained on the basis of $1000$ predictor mean replicates. The outcome is reported in Figure~\ref{fig:boxplots:lin:Gauss} and Table~\ref{tab:lags:means:stds:lin:Gauss}, which for this model indicate a somewhat even more robust performance of our estimator with respect to the lag size; indeed, more or less all the lag sizes in the interval $\intvect{12}{22}$ yield negligible biases, with only a slight increase of variance for the larger ones. According to Table~\ref{tab:lags:means:stds:lin:Gauss}, $\lambda = 18$ yields the minimal bias. As in the previous example, The CLE (corresponding to the lag size $\lambda = 600$) exhibits very unstable performance due to the relatively high $n$-to-$\N$ ratio, with at least $75\%$ of its estimates falling below the reference value. 

\begin{figure}[H] 
\centering
\includegraphics[width=0.8\textwidth]{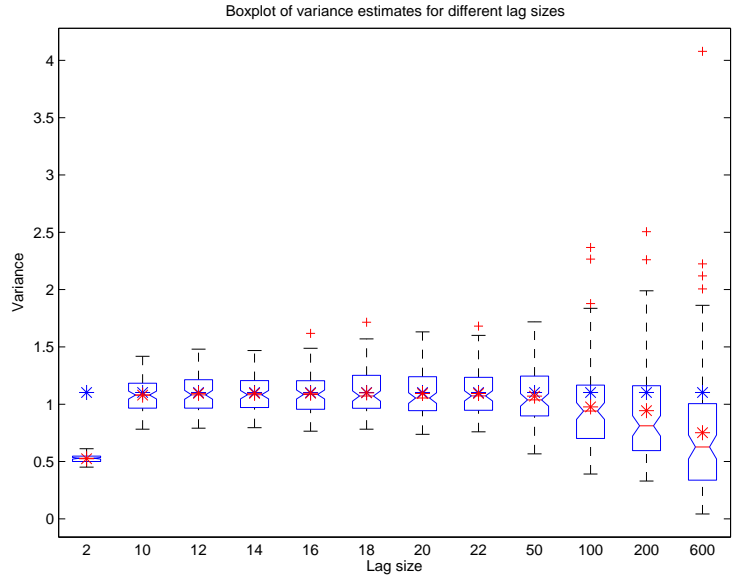} 
\caption{Estimated asymptotic variances of the particle predictor mean at time $600$ in the linear Gaussian model \eqref{eq:lin:Gauss:model}. The particle population size is $\N = 4000$. The boxes are based on $100$ replicates of Algorithm~\ref{alg:fixed-lag:SMC} and correspond to the different lags $\lag \in \{2, 10, 12, 14, 16, 18, 20, 22, 50, 100, 200, 600\}$. The reference value $1.102$, which is indicated by blue-colored asterisks $(\ast)$, was obtained as the sample variance (scaled by the number of particles) of $1000$ independent replicates of the particle predictor mean at time $600$ (again, Algorithm~\ref{alg:fixed-lag:SMC} used $\N = 4000$ particles). Red-colored asterisks indicate average estimates of the boxes.}
\label{fig:boxplots:lin:Gauss}
\end{figure}

\begin{table}[H] 
\begin{center}
\begin{tabular}{c|c|c} \toprule
    $\lambda$ & Mean & St. dev. \\ \midrule 
    $2$     & $.524$   & $.035$ \\
    $10$   & $1.080$ & $.157$ \\ 
    $12$   & $1.095$ & $.163$ \\
    $14$   & $1.095$Ê& $.162$ \\
    $16$   & $1.096$ & $.175$ \\
    $18$   & $1.099$ & $.190$ \\
    $20$   & $1.094$ & $.198$ \\
    $22$   & $1.093$ & $.202$ \\
    $50$   & $1.071$ & $.246$ \\
    $100$ & $.976$   & $.370$ \\
    $200$ & $.944$   & $.471$ \\
    $600$ & $.751$   & $.593$ \\
\bottomrule
\end{tabular} 
\end{center}
\caption{Means and standard deviations of the variance estimates reported in Figure~\ref{fig:boxplots:lin:Gauss}. The reference value is $1.102$} 
\label{tab:lags:means:stds:lin:Gauss}
\end{table}

Finally, using the lag $\lag = 18$ extracted from the previous simulation, Algorithm~\ref{alg:fixed-lag:SMC} was re-run $150$ times on the same observation record $\chunk{y}{0}{n - 1}$, $n = 600$, each run producing a sequence $\{ \predpart[\chunk{y}{0}{m - 1}](\operatorname{id}) \}_{m = 0}^n$ of particle predictor means, a sequence $\{ \varest{\chunk{y}{0}{m - 1}}[\lagtime{m}{\lag}](\operatorname{id}) \}_{m = 0}^n$ of fixed-lag variance estimates, and associated approximate $95\%$ confidence intervals $\{ I_m \}_{m = 0}^n$, each interval given by    
\begin{equation} \label{eq:confidence:bound}
I_m = \left(\predpart[\chunk{y}{0}{m - 1}] \operatorname{id} \pm \lambda_{.025} \frac{\varest{\chunk{y}{0}{m - 1}}[\lagtime{m}{\lag}](\operatorname{id})}{\sqrt{\N}} \right), 
\end{equation}
where $\lambda_{.025}$ denotes the $2.5\%$ quantile of the standard Gaussian distribution. As before, $\N$ was set to $4000$. In the case of effective variance estimation, one may expect each $I_m$ to fail to cover the true predicted mean $\pred[\chunk{y}{0}{m - 1}](\operatorname{id})$ with a probability close to $5\%$. Since we are in the setting of a linear Gaussian model, the exact predictor means $\{ \pred[\chunk{y}{0}{m - 1}](\operatorname{id}) \}_{m = 0}^n$ are accessible through Kalman prediction, and we are thus able to assess the failure rates of the confidence intervals at different time points. Such failure rates are reported in Figure~\ref{fig:stemplot:FRs}, and the red-dashed line indicates the perfect rate $5\%$. For readability, only every $10^{\mathrm{th}}$ time point is reported. Appealingly, it is clear that the rates fluctuate constantly around $5\%$, without any notable time dependence. This re-confirms the numeric stability of our fixed-lag variance estimator. The average failure rate across all $600$ time points was $5.5\%$. The fact that the average failure rate is slightly above the perfect rate $5\%$ is in line with the fact that the bias of our estimator is always positive, as underestimation of variance leads to more narrow confidence bounds and, consequently, higher failure rates. One way of hedging against underestimation of variance could be to replace Gaussian quantiles by the quantiles of some Student's $t$-distribution with a moderate number of degrees of freedom.  

\begin{figure}[H] 
\centering
\includegraphics[width=0.8\textwidth]{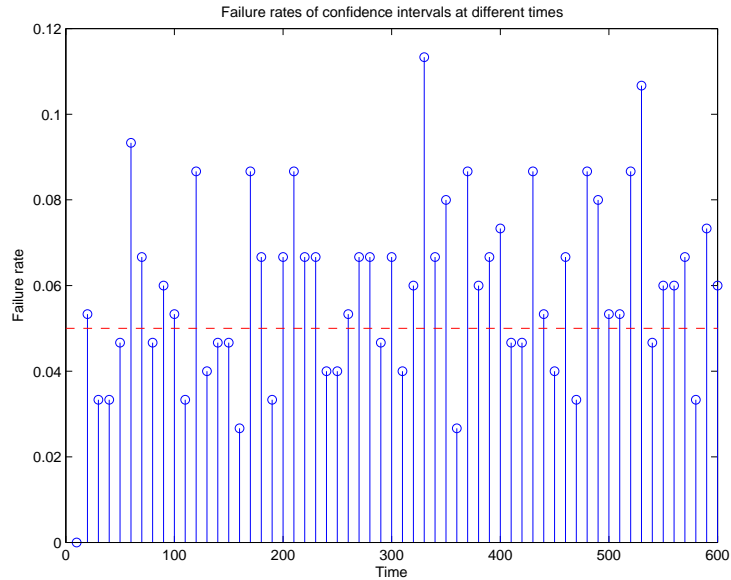} 
\caption{Failure rates of confidence bounds \eqref{eq:confidence:bound} at time points $m \in \{10, 20, 30, \ldots, 600\}$. The red-dashed line indicates the perfect rate $5\%$. The failure rate estimates are based on $150$ runs of Algorithm~\ref{alg:fixed-lag:SMC} with $\N = 4000$ particles.}
\label{fig:stemplot:FRs}
\end{figure}

\section{Conclusion}
\label{sec:conclusion}

The estimator of the SMC asymptotic variance that we propose is a natural modification of the CLE introduced in \cite{chan:lai:2013}. As in \cite{kitagawa:sato:2001,olsson:cappe:douc:moulines:2006}, the main idea is to reduce the degree of genealogical tracing, which has a devastating effect on the CLE's numerical stability, at the cost of a small bias, which may be controlled using the forgetting properties of the model. That this measure stabilises numerically the estimator in the long run is confirmed by our theoretical results in Section~\ref{sec:theoretical:results}, which are obtained under---what we believe---minimal model assumptions being satisfied also for many models with possibly non-compact state space. The fact that the bias can be shown to decrease geometrically fast as the lag increases indicates that tight control of the bias is possible also for moderately large particle sample sizes. This is approved by our numerical experiments in Section~\ref{sec:numerical:study}, which report, in the examples under consideration, a negligible bias already for some thousands of particles.

Needless to say, the success of our approach depends highly on the interplay between the forgetting properties of the model, the particle sample size, and the choice of the lag. Adaptive lag design is hence a natural direction for future research. Moreover, as our estimator provides numerically stable estimates of the asymptotic variance, it should be highly useful for online SMC sample size adaptation. Here one natural approach could be to estimate the variance of the next time step using a part of the particle population (\emph{pilot sampling}) and then ``refuel'' the particle system at time steps of high variance (here the techniques developed in \cite{lee:whiteley:2016}, where the authors consider the SMC sample allocation problem in the batch mode, should be useful for the theoretical analysis). Finally, casting, using the results obtained in \cite{lindsten:shoen:olsson:2011}, our estimator and analysis into the framework of \emph{Rao-Blackwellised SMC algorithms}, should be of high relevance for high-dimensional applications.

\appendix

\section{Proofs}
\label{sec:proofs}
\subsection{Proof of Proposition~\ref{prop:consistency:fixed:lag}}
\label{sec:proof:consistency:fixed:lag}

The proof of Proposition~\ref{prop:consistency:fixed:lag} relies on the machinery developed in \cite{lee:whiteley:2016}, from which we adopt the following definitions. Throughout this section, 
let $n \in \nset$ and $\chunk{z}{0}{n - 1} \in \Zsp^n$ be picked arbitrarily.  
\begin{itemize}
\item Denote by $\Binsp{n} \eqdef \{0, 1\}^{n + 1}$ the space of binary strings of length $n + 1$. The \emph{zero string} of length $n + 1$ is denoted by $\zerostr{n}$ and for $m \in \intvect{0}{n}$, $\unitstr{m}{n}$ denotes a \emph{unit string} of length $n + 1$ with $1$ on position $m$ (with positions indexed from $0$) and zeros everywhere else.  
\item For a given string $\chunk{b}{0}{n} \in \Binsp{n}$, a Markov chain $\{Ê(X_m, X_m') \}_{m = 0}^n$ on $(\Xsp^2, \Xfd^{\varotimes 2})$ is defined as follows. If $b_0 = 0$, then $(X_0, X_0') \sim \init^{\varotimes 2}$; otherwise, if $b_0 = 1$, $X_0' = X_0 \sim \init$ (the initial distribution). After this, if $b_{m + 1} = 0$, $X_{m + 1} \sim \mk(X_m, \cdot)$Ê and $X_{m + 1}' \sim \mk(X_m', \cdot)$ conditionally independently; otherwise, if $b_{m + 1} = 1$, $X_{m + 1}' = X_{m + 1} \sim \mk(X_m, \cdot)$. 
\item With $\E_{\chunk{b}{0}{n}}$ denoting the expectation under the law of $\{Ê(X_m, X_m') \}_{m = 0}^n$, we define, for all $\chunk{b}{0}{n} \in \Binsp{n}$, the measures 
$$
\mu_{\chunk{b}{0}{n}} \langle \chunk{z}{0}{n - 1} \rangle : \Xfd^{\varotimes 2} \ni A \mapsto 
\E_{\chunk{b}{0}{n}} \left[ \1_A(X_n, X_n') \prod_{m = 0}^{n - 1} \pot[z_m](X_m) \pot[z_m](X_m') \right]. 
$$ 
Note that for all $h \in \bmf{\Xfd}$ it holds that $\mu_{\zerostr{n}} \langle \chunk{z}{0}{n - 1} \rangle h^{\varotimes 2} = (\init \uk[ \chunk{z}{0}{n - 1}] h)^2$ and $\mu_{\unitstr{m}{n}} \langle \chunk{z}{0}{n - 1} \rangle h^{\varotimes 2} = \init \uk[\chunk{z}{0}{m - 1}] \1_\Xsp \times \init \uk[\chunk{z}{0}{m - 1}](\uk[\chunk{z}{m}{n - 1}] h)^2$, and defining 
$$
\term[\chunk{z}{0}{n - 1}]{m}{n} : \bmf{\Xfd} \ni h \mapsto \frac{\mu_{\unitstr{m}{n}} \langle \chunk{z}{0}{n - 1} \rangle h^{\varotimes 2} - \mu_{\zerostr{n}} \langle \chunk{z}{0}{n - 1} \rangle h^{\varotimes 2}}{(\init \uk[\chunk{z}{0}{n - 1}] \1_\Xsp)^2}
$$
yields for all $h \in \bmf{\Xfd}$, 
$$
\term[\chunk{z}{0}{n - 1}]{m}{n}(h) = \frac{\pred[\chunk{z}{0}{m - 1}] (\uk[\chunk{z}{m}{n - 1}] h)^2}{(\pred[\chunk{z}{0}{m - 1}] \uk[\chunk{z}{m}{n - 1}] \1_{\Xsp})^2} - (\pred[\chunk{z}{0}{n - 1}] h)^2 
$$
and, consequently, for all $\ell \in \intvect{0}{n}$, 
\begin{equation} \label{eq:variance:alt:expression}
\variance[2]{\chunk{z}{0}{n - 1}}[\ell](h) = \sum_{m = \ell}^n \term[\chunk{z}{0}{n - 1}]{m}{n} (h - \pred[\chunk{z}{0}{n - 1}] h). 
\end{equation}
\item For all $\N \in \nsetpos$, let Ê$\partfd{n} \eqdef \sigma( \{ \epart{0}{i} \}_{i = 1}^\N, \{Ê\epart{m}{i}, \ind{m}{i} \}_{i = 1}^\N ; m \in \intvect{1}{n} )$ be the $\sigma$-field generated by the output of Algorithm~\ref{alg:SMC} during the first $n$ iterations. Conditionally on $\partfd{n}$, a genealogical trace $\chunk{\gen[1]}{0}{n}$ is formed backwards in time by, first, drawing $\gen[1]_n$ uniformly over $\intvect{1}{\N}$ and, second, setting $\gen[1]_m = \ind{m + 1}{\gen[1]_{m + 1}}$ for all $m \in \intvect{0}{n - 1}$. In addition, a parallel trace $\chunk{\gen[2]}{0}{n}$ is formed by, first, drawing $\gen[2]_n$ uniformly over $\intvect{1}{\N}$ and, second, letting $\gen[2]_m = \ind{m + 1}{\gen[2]_{m + 1}}$ if $\gen[2]_{m + 1} \neq \gen[1]_{m + 1}$ or $\gen[2]_m \sim \cat(\{ \wgt{m}{i} / \wgtsum{m} \}_{i = 1}^\N)$ otherwise. 
\end{itemize}

The proof of the following lemma follows closely that of \cite[Lemma~4]{lee:whiteley:2016} and is hence omitted. Define for all $\ell \in \intvect{0}{n}$ and $\chunk{b}{0}{\ell} \in \Binsp{\ell}$, 
$$
\binset{\ell}(\chunk{b}{0}{\ell}) \eqdef \left \{Ê(\chunk{k}{0}{\ell}, \chunk{k'}{0}{\ell}) \in \intvect{1}{\N}^{2(\ell + 1)} : \mbox{for all } \ell' \in \intvect{0}{\ell}, \ k_{\ell'} = k'_{\ell'} \Leftrightarrow b_{\ell'} = 1 \right \}. 
$$
\begin{lemma} \label{lemma:equiv:sets}
For all $\N \in \nsetpos$ and $m \in \intvect{0}{n}$, 
$$
\left \{ \enoch{m}{n}{\gen[1]_n} \neq \enoch{m}{n}{\gen[2]_n} \right \} = \left \{ (\chunk{\gen[1]}{m}{n}, \chunk{\gen[2]}{m}{n}) \in \binset{n - m}(0_{n - m}) \right \}. 
$$ 
\end{lemma}

In addition, define, for all $\N \in \nsetpos$, the measures
\begin{equation} \label{eq:def:part:gamma}
\unpredpart[\chunk{z}{0}{n - 1}] : \Xfd \ni A \mapsto \frac{1}{\N^{n + 1}} \left( \prod_{m = 0}^{n - 1} \wgtsum{m} \right) \sum_{i = 1}^\N \1_A(\epart{n}{i})
\end{equation}
and 
\begin{multline} \label{eq:mu:meas}
\mumeas{\N, \chunk{b}{0}{n}}{\chunk{z}{0}{n - 1}} :  \Xfd^{\varotimes 2} \ni A \mapsto 
\N^{\#_1(\chunk{b}{0}{n})} \left( \frac{\N}{\N - 1} \right)^{\#_0(\chunk{b}{0}{n})}
\left( \unpredpart[\chunk{z}{0}{n - 1}] \1_\Xsp \right)^2 \\ 
\times \E \left[ \1_A \left( \epart{n}{\gen[1]_n}, \epart{n}{\gen[2]_n} \right) \1 \left \{Ê(\chunk{\gen[1]}{0}{n}, \chunk{\gen[2]}{0}{n}) \in \binset{n}(\chunk{b}{0}{n}) \right \} \mid \partfd{n} \right],  
\end{multline}
where $\#_1(\chunk{b}{0}{n}) \eqdef \sum_{m = 0}^n b_m$ and $\#_0(\chunk{b}{0}{n}) \eqdef n + 1 - \#_1(\chunk{b}{0}{n})$ denote the numbers of ones and zeros in $\chunk{b}{0}{n}$, respectively. Note that \eqref{eq:def:part:gamma} implies that for all $h \in \bmf{\Xfd}$, $\predpart[\chunk{z}{0}{n - 1}] h = \unpredpart[\chunk{z}{0}{n - 1}] h / \unpredpart[\chunk{z}{0}{n - 1}] \1_\Xsp$. 

The following lemma, where first part is established in \cite[Theorem~2]{lee:whiteley:2016} and the last part is a standard result (see, e.g, \cite{douc:moulines:2008} for results on the weak consistency of SMC), will be instrumental. 

\begin{lemma} \label{lemma:mu:convergence}
For all $\chunk{b}{0}{n} \in \Binsp{n}$ and $h \in \bmf{\Xfd^{\varotimes 2}}$, as $\N \rightarrow \infty$, 
$$
\mumeas{\N, \chunk{b}{0}{n}}{\chunk{z}{0}{n - 1}} h \plim \mumeas{\chunk{b}{0}{n}}{\chunk{z}{0}{n - 1}} h. 
$$
In addition, for all $h \in \bmf{\Xfd}$, 
$$
\unpredpart[\chunk{z}{0}{n - 1}] h \plim \init \uk[\chunk{z}{0}{n - 1}] h. 
$$
\end{lemma}

\begin{proof}[Proof of Proposition~\ref{prop:consistency:fixed:lag}]
Fix $\ell \in \intvect{0}{n}$ and define for all $\N \in \nsetpos$,
\begin{multline*}
\varphi_{\N, \ell} \langle \chunk{z}{0}{n - 1} \rangle : \Xfd^{\varotimes 2} \ni A \mapsto  \left( \unpredpart[\chunk{z}{0}{n - 1}] \1_\Xsp \right)^2 \\Ê
\times \E \left[ \1_A \left( \epart{n}{\gen[1]_n}, \epart{n}{\gen[2]_n} \right) \1 \left \{Ê(\chunk{\gen[1]}{\ell}{n}, \chunk{\gen[2]}{\ell}{n}) \in \binset{n - \ell}(0_{n - \ell}) \right \} \mid \partfd{n} \right]. 
\end{multline*}
First, note that by Lemma~\ref{lemma:equiv:sets}, since $\gen[1]_n$ and $\gen[2]_n$ are conditionally independent and uniformly distributed over $\intvect{1}{\N}$, for all $h \in \bmf{\Xfd}$,  
\begin{equation} \label{eq:estimator:alt:expression}
\begin{split} 
\frac{1}{\N^2} \sum_{i = 1}^\N \left( \sum_{j : \enoch{\ell}{n}{j} = i} h(\epart{n}{j}) \right)^2 
&= \frac{1}{\N^2} \sum_{(i, j) : \enoch{\ell}{n}{i} = \enoch{\ell}{n}{j}} h(\epart{n}{i}) h(\epart{n}{j}) \\
&= (\predpart[\chunk{z}{0}{n - 1}] h)^2 - \frac{1}{\N^2} \sum_{(i, j) : \enoch{\ell}{n}{i} \neq \enoch{\ell}{n}{j}} h(\epart{n}{i}) h(\epart{n}{j}) \\
&= (\predpart[\chunk{z}{0}{n - 1}] h)^2 - \frac{\varphi_{\N, \ell} \langle \chunk{z}{0}{n - 1} \rangle h^{\varotimes 2}}{(\unpredpart[\chunk{z}{0}{n - 1}] \1_\Xsp)^2}.  
\end{split}
\end{equation}
It is hence enough to prove that for all $\N \in \nsetpos$ and $h \in \bmf{\Xfd}$, 
\begin{multline} \label{eq:sufficient:condition}
\N \left \{ (\unpredpart[\chunk{z}{0}{n - 1}] h)^2 - \varphi_{\N, \ell} \langle \chunk{z}{0}{n - 1} \rangle h^{\varotimes 2} \right\} \\Ê
= 
\sum_{m = \ell}^n \left( \mumeas{\N, \unitstr{m}{n}}{\chunk{z}{0}{n - 1}} h^{\varotimes 2} - \mumeas{\N, \zerostr{n}}{\chunk{z}{0}{n - 1}} h^{\varotimes 2} \right) + (n - \ell + 1) (\unpredpart[\chunk{z}{0}{n - 1}] h)^2 \\Ê
+ \| h \|_\infty^2 \ordo(\N^{-1}),
\end{multline}
where the $\ordo(\N^{-2})$ term does not depend on $h$; indeed, along the lines of the proof of \cite[Theorem~1]{lee:whiteley:2016}, Lemma~\ref{lemma:mu:convergence} implies that for all $\chunk{b}{0}{n} \in \Binsp{n}$, as $\N \rightarrow \infty$, 
$$
\mumeas{\N, \chunk{b}{0}{n}}{\chunk{z}{0}{n - 1}} \{Êh - \predpart[\chunk{z}{0}{n - 1}] h \}^{\varotimes 2} \plim \mumeas{\chunk{b}{0}{n}}{\chunk{z}{0}{n - 1}} \{Êh - \pred[\chunk{z}{0}{n - 1}] h \}^{\varotimes 2}, 
$$
and \eqref{eq:sufficient:condition} hence yields, with $\ell = \lagtime{n}{\lag}$ Êand when combined with \eqref{eq:estimator:alt:expression} and \eqref{eq:variance:alt:expression}, again as $\N \rightarrow \infty$, 
\begin{multline} \label{eq:critical:identity}
\varest[2]{\chunk{z}{0}{n - 1}}[\lag](h) \\
= \sum_{m = \lagtime{n}{\lag}}^n \frac{\mumeas{\N, \unitstr{m}{n}}{\chunk{z}{0}{n - 1}} \{Êh - \predpart[\chunk{z}{0}{n - 1}] h \}^{\varotimes 2} - \mumeas{\N, 0_n}{\chunk{z}{0}{n - 1}} \{Êh - \predpart[\chunk{z}{0}{n - 1}] h \}^{\varotimes 2}}{(\unpredpart[\chunk{z}{0}{n - 1}] \1_\Xsp)^2} \\
+ \| h \|_\infty^2 \ordo(\N^{-1}) \plim \variance[2]{\chunk{z}{0}{n - 1}}[\lagtime{n}{\lag}](h). 
\end{multline}
In order to establish \eqref{eq:sufficient:condition}, write, using that $\gen[1]_n$ and $\gen[2]_n$ are, given $ \partfd{n}$, conditionally independent and uniformly distributed over $\intvect{1}{\N}$,   
\begin{align} \label{eq:estimator:difference:form}
\lefteqn{(\unpredpart[\chunk{z}{0}{n - 1}] h)^2 - \varphi_{\N, \ell} \langle \chunk{z}{0}{n - 1} \rangle h^{\varotimes 2}} \nonumber \\ 
&= (\unpredpart[\chunk{z}{0}{n - 1}] \1_\Xsp)^2 \sum_{\chunk{b}{0}{n} \in \Binsp{n}} \E \left[ h\big( \epart{n}{\gen[1]_n} \big) h \big( \epart{n}{\gen[2]_n} \big) \1 \left \{Ê(\chunk{\gen[1]}{0}{n}, \chunk{\gen[2]}{0}{n}) \in \binset{n}(\chunk{b}{0}{n}) \right \} \mid \partfd{n} \right] \nonumber \\ 
&- (\unpredpart[\chunk{z}{0}{n - 1}] \1_\Xsp)^2 \sum_{\chunk{b}{0}{n} \in \Binsp{n} : \chunk{b}{\ell}{n} = 0_{n - \ell}} \E \left[ h\big( \epart{n}{\gen[1]_n} \big) h \big( \epart{n}{\gen[2]_n} \big) \1 \left \{Ê(\chunk{\gen[1]}{0}{n}, \chunk{\gen[2]}{0}{n}) \in \binset{n}(\chunk{b}{0}{n}) \right \} \mid \partfd{n} \right]
\end{align}
and note that, by definition \eqref{eq:mu:meas}, 
\begin{align} \label{eq:estimator:first:term}
\lefteqn{(\unpredpart[\chunk{z}{0}{n - 1}] \1_\Xsp)^2 \sum_{\chunk{b}{0}{n} \in \Binsp{n}} \E \left[ h \big( \epart{n}{\gen[1]_n} \big) h \big( \epart{n}{\gen[2]_n} \big) \1 \left \{Ê(\chunk{\gen[1]}{0}{n}, \chunk{\gen[2]}{0}{n}) \in \binset{n}(\chunk{b}{0}{n}) \right \} \mid \partfd{n} \right]} \nonumber \\Ê
&= \frac{1}{\N} \sum_{m = 0}^n \mumeas{\N, \unitstr{m}{n}}{\chunk{z}{0}{n - 1}} h^{\varotimes 2} + \left( 1 - \frac{1}{\N} \right)^{n + 1} \mumeas{\N, 0_n}{\chunk{z}{0}{n - 1}} h^{\varotimes 2} + \| h \|_\infty^2 \ordo(\N^{-2}) \nonumber \\
&= \frac{1}{\N} \sum_{m = 0}^n \left( \mumeas{\N, \unitstr{m}{n}}{\chunk{z}{0}{n - 1}} h^{\varotimes 2}  - \mumeas{\N, \zerostr{n}}{\chunk{z}{0}{n - 1}} h^{\varotimes 2} \right) + \mumeas{\N, 0_n}{\chunk{z}{0}{n - 1}} h^{\varotimes 2} + \| h \|_\infty^2 \ordo(\N^{-2}). 
\end{align}
Similarly, 
\begin{multline} \label{eq:estimator:second:term}
\lefteqn{(\unpredpart[\chunk{z}{0}{n - 1}] \1_\Xsp)^2 \sum_{\chunk{b}{0}{n} \in \Binsp{n} : \chunk{b}{\ell}{n} = 0_{n - \ell}} \E \left[ h\big( \epart{n}{\gen[1]_n} \big) h \big( \epart{n}{\gen[2]_n} \big) \1 \left \{Ê(\chunk{\gen[1]}{0}{n}, \chunk{\gen[2]}{0}{n}) \in \binset{n}(\chunk{b}{0}{n}) \right \} \mid \partfd{n} \right]} \\
= \frac{1}{\N} \sum_{m = 0}^{\ell - 1} \left( \mumeas{\N, \unitstr{m}{n}}{\chunk{z}{0}{n - 1}} h^{\varotimes 2}  - \mumeas{\N, \zerostr{n}}{\chunk{z}{0}{n - 1}} h^{\varotimes 2} \right) + \mumeas{\N, 0_n}{\chunk{z}{0}{n - 1}} h^{\varotimes 2} \\Ê
- \frac{n - \ell + 1}{\N} \mumeas{\N, 0_n}{\chunk{z}{0}{n - 1}} h^{\varotimes 2} + \| h \|_\infty^2 \ordo(\N^{-2}), 
\end{multline}
and combining \eqref{eq:estimator:difference:form}, \eqref{eq:estimator:first:term}, \eqref{eq:estimator:second:term}, and the fact that $\mumeas{0_n}{\chunk{z}{0}{n - 1}} h^{\varotimes 2} = (\init \uk[\chunk{z}{0}{n - 1}] h)^2$ yields \eqref{eq:sufficient:condition}. This completes the proof. 
\end{proof}

\subsection{Proof of Theorem~\ref{thm:tightness:bias}}

In the proof of Theorem~\ref{thm:tightness:bias}, the asymptotic bias is bounded using the time-shift approach taken in \cite[Theorem~10]{douc:moulines:olsson:2014}. Even though the theoretical analysis provided in \cite{douc:moulines:olsson:2014} is cast into the framework of general state-space models, it never makes use of the fact that $\pot$ is a normalised transition density. As stressed in \cite[Remark~1]{douc:moulines:olsson:2014}, it is hence directly applicable to the framework of randomly perturbed Feynman-Kac models in Section~\ref{sec:Feynman:Kac:models}. 

\begin{proof}
Pick arbitrarily $n \in \nset$, $\lag \in \nset$, $h \in \bmf{\Xfd}$, and $\init \in \mdr(D, r)$. 
By defining, for all $m \in \nset$, $\ell \in \intvect{0}{m - 1}$, $\chunk{z}{\ell}{m - 1} \in \Zsp^{m - \ell}$, and measures $(\mu, \mu') \in \probmeas{\Xfd}^2$, 
\begin{multline} \label{eq:def-Delta}
\Delta_{\mu, \mu'} \langle \chunk{z}{\ell}{m - 1} \rangle : \bmf{\Xfd}^2 \ni (h, h') \mapsto 
\mu \uk[\chunk{z}{\ell}{m - 1}] h \times \mu' \uk[\chunk{z}{\ell}{m - 1}] h'
\\ - \mu \uk[\chunk{z}{\ell}{m - 1}] h' \times \mu' \uk[\chunk{z}{\ell}{m - 1}] h, 
\end{multline}
we may, using the identity 
$$
\pred[\chunk{\per}{0}{m - 1}] \uk[\chunk{\per}{m}{n - 1}] \1_{\Xsp} = \frac{\init\uk[\chunk{\per}{0}{n - 1}] \1_{\Xsp}}{\init\uk[\chunk{\per}{0}{m - 1}] \1_{\Xsp}} = \prod_{\ell = m}^{n - 1} \frac{\init\uk[\chunk{\per}{0}{\ell}] \1_{\Xsp}}{\init\uk[\chunk{\per}{0}{\ell - 1}] \1_{\Xsp}} = \prod_{\ell = m}^{n - 1} \pred[\chunk{\per}{0}{\ell - 1}]  \pot[\per_\ell],
$$
for all $m \in \intvect{0}{n}$,
write the asymptotic bias at time $n$ as
\begin{equation} \label{bias:alternative:form}
\bias{\chunk{\per}{0}{n - 1}}{\lag}(h) 
= \sum_{m = 0}^{\lagtime{n}{\lag} - 1} \int \pred[\chunk{\per}{0}{m - 1}](\rmd x) \left( \frac{\DDelta{\delta_x,\pred[\chunk{\per}{0}{m - 1}]}{\chunk{\per}{m}{n - 1}}{h, \1_{\Xsp}}}
{[\prod_{\ell = m}^{n - 1} \pred[\chunk{\per}{0}{\ell - 1}]  \pot[\per_\ell]]^2} \right)^2. 
\end{equation}
In addition, under the assumptions of the theorem, \cite[Proposition~1]{douc:moulines:2012} provides the existence of a function $\pi: \Zsp^\infty \to \rset$ such that for all initial distributions $\init \in \mdr(D, r)$,
$$
\lim_{m \to \infty} \pred[\chunk{\per}{-m}{- 1}] \pot[\per_0] = \limitfunc{\chunk{\per}{- \infty}{0}}, \quad \prob\mbox{-a.s.}
$$
Since the perturbations $\{ \per_n \}_{n \in \zset}$ are stationary, the distribution of $\bias{\chunk{\per}{0}{n - 1}}{\lag}(h)$ coincides with that of the time-shifted bias $\bias{\chunk{\per}{-n}{- 1}}{\lag}(h)$, and a key step in the present proof is to express, via \eqref{bias:alternative:form}, the latter as 
\begin{multline*}
\bias{\chunk{\per}{-n}{- 1}}{\lag}(h) = \\Ê
\sum_{m = 0}^{\lagtime{n}{\lag} - 1} \int \pred[\chunk{\per}{-n}{- n + \lagtime{n}{\lag} - m - 2}](\rmd x) \left( \frac{\DDelta{\delta_x,\pred[\chunk{\per}{-n}{- n + \lagtime{n}{\lag} - m - 2}]}{\chunk{\per}{- n + \lagtime{n}{\lag} - m - 1}{-1}}{h, \1_{\Xsp}}}{[\prod_{\ell = 1}^{n - \lagtime{n}{\lag} + m + 1} \pred[\chunk{\per}{-n}{- \ell - 1}]  \pot[\per_{- \ell}]]^2 } \right)^2.  
\end{multline*}
We hence obtain the bound 
\begin{equation} \label{eq:fundamental:bias:bound}
\bias{\chunk{\per}{-n}{-1}}{\lag}(h) \leq A_n \times B_n,
\end{equation}
where
\begin{align*} 
A_n &\eqdef \left(\sup_{(k, m) \in \zset^2 : \, -n \leq k \leq m} \prod_{\ell = k}^{m - 1} \frac{\limitfunc{\chunk{\per}{-\infty}{\ell}}}{\pred[\chunk{\per}{-n}{\ell - 1}] \pot[\per_\ell]} \right)^4, \\ 
B_n &\eqdef \sum_{m = 0}^{\lagtime{n}{\lag} - 1} \left( \frac{\sup_{x \in \Xsp} |\DDelta{\delta_x, \pred[\chunk{\per}{-n}{- n + \lagtime{n}{\lag} - m - 2}]}{\chunk{\per}{- n + \lagtime{n}{\lag} - m - 1}{-1}}{h, \1_\Xsp}|}{[\prod_{\ell = 1}^{n - \lagtime{n}{\lag} + m + 1} \limitfunc{\chunk{\per}{-\infty}{- \ell}}]^2} \right)^2.
\end{align*}
To bound uniformly the sequence $\{ B_n \}_{n \in \nset}$, decompose each term according to 
\begin{multline} \label{eq:termwise:decomposition:Bn}
\frac{\sup_{x \in \Xsp} |\DDelta{\delta_x, \pred[\chunk{\per}{-n}{- n + \lagtime{n}{\lag} - m - 2}]}{\chunk{\per}{- n + \lagtime{n}{\lag} - m - 1}{-1}}{h, \1_\Xsp}|}{[\prod_{\ell = 1}^{n - \lagtime{n}{\lag} + m + 1} \limitfunc{\chunk{\per}{-\infty}{- \ell}}]^2} \\
= \left( \frac{\| \uk[\chunk{\per}{- n + \lagtime{n}{\lag} - m - 1}{-1}] \1_\Xsp \|_\infty}{\prod_{\ell = 1}^{n - \lagtime{n}{\lag} + m + 1} \limitfunc{\chunk{\per}{-\infty}{- \ell}}} \right)^2 \\
\times \frac{\sup_{x \in \Xsp} |\DDelta{\delta_x, \pred[\chunk{\per}{-n}{- n + \lagtime{n}{\lag} - m - 2}]}{\chunk{\per}{- n + \lagtime{n}{\lag} - m - 1}{-1}}{h, \1_\Xsp}|}{\| \uk[\chunk{\per}{- n + \lagtime{n}{\lag} - m - 1}{-1}] \1_\Xsp \|_\infty^2}. 
\end{multline}
We consider separately the two factors of \eqref{eq:termwise:decomposition:Bn}. First,
\begin{equation} \label{eq:first:factor:Bn:deomposition}
\left( \frac{\| \uk[\chunk{\per}{- n + \lagtime{n}{\lag} - m - 1}{-1}] \1_\Xsp \|_\infty}{\prod_{\ell = 1}^{n - \lagtime{n}{\lag} + m + 1} \limitfunc{\chunk{\per}{-\infty}{- \ell}}} \right)^2 = \exp\{Ê(n - \lagtime{n}{\lag} + m + 1) \varepsilon_{n - \lagtime{n}{\lag} + m + 1} \}, 
\end{equation}
with
$$
\varepsilon_k \eqdef \frac{2}{k} \left( \ln \| \uk[\chunk{\per}{-k}{-1}] \1_\Xsp \|_\infty - \sum_{\ell = 1}^k \ln  \limitfunc{\chunk{\per}{-\infty}{- \ell}} \right)
$$
being independent of $\init$ for all $k \in \nsetpos$.  
By \cite[Lemma~17]{douc:moulines:olsson:2014}, $\varepsilon_k \to 0$, $\prob$-a.s., as $k \to \infty$, which implies that \eqref{eq:first:factor:Bn:deomposition} grows at most subgeometrically fast with $m$.  In addition, by \cite[Proposition~16(iii)]{douc:moulines:olsson:2014} there exists a constant $\rho \in (0, 1)$ and a $\prob$-a.s. finite random variable $D$ such that for all $n$ and $m$, all $h \in \bmf{\Xfd}$, and all $\init \in \probmeas{\Xfd}$, $\prob$-a.s,
$$
\frac{\sup_{x \in \Xsp} |\DDelta{\delta_x, \pred[\chunk{\per}{-n}{- n + \lagtime{n}{\lag} - m - 2}]}{\chunk{\per}{- n + \lagtime{n}{\lag} - m - 1}{-1}}{h, \1_\Xsp}|}{\| \uk[\chunk{\per}{- n + \lagtime{n}{\lag} - m - 1}{-1}] \1_\Xsp \|_\infty^2} \leq D \rho^{n - \lagtime{n}{\lag} + m + 1} \| h \|_\infty. 
$$
Thus, $\prob$-a.s,
$$
B_n \leq D^2 \| h \|_\infty^2 \sum_{m = 0}^{\lagtime{n}{\lag} - 1} \rho^{2(n - \lagtime{n}{\lag} + m + 1)} \exp\{Ê2 (n - \lagtime{n}{\lag} + m + 1) \varepsilon_{n - \lagtime{n}{\lag} + m + 1} \}.  
$$
If $n \leq \lag$, then $\lagtime{n}{\lag} = 0$, and the bias vanishes. Thus, we assume in the following that $\lag < n$, which means that $\lagtime{n}{\lag} = n - \lag$. Then, by the Cauchy-Schwartz inequality, $\prob$-a.s,
\begin{align}
B_n &\leq D^2  \| h \|_\infty^2 \sum_{m = \lag + 1}^\infty \rho^{2m} \exp(2m \varepsilon_m) \nonumber \\
&\leq D^2 \| h \|_\infty^2 \left( \sum_{m = \lag + 1}^\infty \rho^{2 m} \right)^{1/2} \left( \sum_{m = \lag + 1}^\infty \rho^{2 m} \exp(4 m \varepsilon_m) \right)^{1/2} \nonumber \\
&\leq D^2 \| h \|_\infty^2 \rho^{\lag + 1} \left( \sum_{m = 0}^\infty \rho^{2 m} \right)^{1/2} \left( \sum_{m = 0}^\infty \rho^{2 m} \exp(4 m \varepsilon_m) \right)^{1/2} \nonumber \\ 
&= C \| h \|_\infty^2 \rho^{\lag + 1}, \label{eq:bias:cauchy:bound}
\end{align}
where random variable  
$$
C \eqdef D^2 \left( \sum_{m = 0}^\infty \rho^{2 m} \right)^{1/2} \left( \sum_{m = 0}^\infty \rho^{2 m} \exp(4 m \varepsilon_m) \right)^{1/2}
$$
is $\prob$-a.s. finite and independent of $\lag$, $h$, and $\init$. For $c \in \rset_+$, write, using \eqref{eq:fundamental:bias:bound} and  \eqref{eq:bias:cauchy:bound}, 
$$
\prob \left( \frac{\bias{\chunk{\per}{0}{n - 1}}{\lag}(h)}{\rate^{\lambda + 1} \| h \|_{\infty}^2} > c \right) = \prob \left( \frac{\bias{\chunk{\per}{-n}{- 1}}{\lag}(h)}{\rate^{\lambda + 1} \| h \|_{\infty}^2} > c \right) \leq \prob \left( A_n C > c \right),  
$$
where the probability on the right hand side is again independent of $\lag$, $h$, or $\init$. Thus, using the stationarity of $\{ A_n \}_{n \in \nset}$, 
$$
\prob \left( \frac{\bias{\chunk{\per}{0}{n - 1}}{\lag}(h)}{\rate^{\lambda + 1} \| h \|_{\infty}^2} > c \right) \leq \prob \left( A_0 > c^{1/2} \right)  + \prob \left( C > c^{1/2} \right), 
$$ 
where the right hand side does not depend on $n$. Now, the $\prob$-a.s. finiteness of $A_0$Ê was established as a part of the proof of \cite[Theorem~10]{douc:moulines:olsson:2014}. Consequently, as also $C$ is $\prob$-a.s. finite, there exists, for all $k \in \nsetpos$, a constant $c_k \in \rset_+$, independent of $\lag$, $h$, and $\init$, such that the probabilities $\prob( A_0 > c^{1/2}_k)$Ê and $\prob( C > c^{1/2}_k)$ are both bounded by $1/(2 k)$. This  completes the proof. 
\end{proof}


\subsection{Proof of Proposition~\ref{prop:consistency:fixed:lag:updated:case}} 

\begin{proof}
The proof consists mainly in combining some of the equalities in the proof of Proposition~\ref{prop:consistency:fixed:lag} with the identity     
\begin{multline} \label{eq:function:prod:identity}
\{Ê\pot[z_n] (h - \filtpart[\chunk{z}{0}{n}] h) \}^{\varotimes 2} = (\pot[z_n] h)^{\varotimes 2} - \{ (\pot[z_n] h) \varotimes \pot[z_n] \} \filtpart[\chunk{z}{0}{n}] h \\Ê
- \{ \pot[z_n] \varotimes (\pot[z_n] h) \} \filtpart[\chunk{z}{0}{n}] h + \pot[z_n]^{\varotimes 2} (\filtpart[\chunk{z}{0}{n}] h)^2. 
\end{multline}
More specifically, as it holds that 
\begin{equation} \label{eq:zero:identity}
\predpart[\chunk{z}{0}{n - 1}](\pot[z_n] \{ h - \filtpart[\chunk{z}{0}{n}] h \})= 0,
\end{equation}   
by reusing the equality in \eqref{eq:critical:identity},
\begin{multline} \label{eq:filtering:variance:numerator}
\varest[2]{\chunk{z}{0}{n - 1}}[\lambda](\pot[z_n] \{ h - \filtpart[\chunk{z}{0}{n}] h\}) = \\
\sum_{m = \lagtime{n}{\lag}}^n \frac{\mumeas{\N, \unitstr{m}{n}}{\chunk{z}{0}{n - 1}} \{Ê\pot[z_n] (h - \filtpart[\chunk{z}{0}{n}] h) \}^{\varotimes 2} - \mumeas{\N, 0_n}{\chunk{z}{0}{n - 1}} \{Ê\pot[z_n] (h - \filtpart[\chunk{z}{0}{n}] h) \}^{\varotimes 2}}{(\unpredpart[\chunk{z}{0}{n - 1}] \1_\Xsp)^2} \\
+ \| \pot[z_n] \|_\infty^2 \| h \|_\infty^2 \ordo(\N^{-1}). 
\end{multline}
Now, write, using \eqref{eq:function:prod:identity}, for all $\chunk{b}{0}{n} \in \Binsp{n}$, 
\begin{multline} \label{eq:termwise:decomposition:updated:case}
\mumeas{\N, \chunk{b}{0}{n}}{\chunk{z}{0}{n - 1}} \{Ê\pot[z_n] (h - \filtpart[\chunk{z}{0}{n}] h) \}^{\varotimes 2} =
\mumeas{\N, \chunk{b}{0}{n}}{\chunk{z}{0}{n - 1}} (\pot[z_n] h)^{\varotimes 2} \\
- \mumeas{\N, \chunk{b}{0}{n}}{\chunk{z}{0}{n - 1}} \{ (\pot[z_n] h) \varotimes \pot[z_n] \} \filtpart[\chunk{z}{0}{n}] h 
- \mumeas{\N, \chunk{b}{0}{n}}{\chunk{z}{0}{n - 1}} \{ \pot[z_n] \varotimes (\pot[z_n] h) \} \filtpart[\chunk{z}{0}{n}] h \\
+ \mumeas{\N, \chunk{b}{0}{n}}{\chunk{z}{0}{n - 1}} \pot[z_n]^{\varotimes 2} (\filtpart[\chunk{z}{0}{n}] h)^2. 
\end{multline}
Applying Lemma~\ref{lemma:mu:convergence} to each term of \eqref{eq:termwise:decomposition:updated:case} (note that the second part of Lemma~\ref{lemma:mu:convergence} implies the consistency of the updated particle measures, as 
\begin{equation} \label{eq:particle:filter:consistency}
\filtpart[\chunk{z}{0}{n}] h = \frac{\predpart[\chunk{z}{0}{n - 1}] (\pot[z_n] h)}{\predpart[\chunk{z}{0}{n - 1}] \pot[z_n]} \plim \frac{\init \uk[\chunk{z}{0}{n - 1}] (\pot[z_n] h)}{\init \uk[\chunk{z}{0}{n - 1}] \pot[z_n]} = \filt[\chunk{z}{0}{n}] h
\end{equation} 
when $\N$ tends to infinity, a now classical result) yields, using again \eqref{eq:function:prod:identity}, 
\begin{multline*}
\mumeas{\N, \chunk{b}{0}{n}}{\chunk{z}{0}{n - 1}} \{Ê\pot[z_n] (h - \filtpart[\chunk{z}{0}{n}] h) \}^{\varotimes 2} \plim 
\mumeas{\chunk{b}{0}{n}}{\chunk{z}{0}{n - 1}} (\pot[z_n] h)^{\varotimes 2} \\
- \mumeas{\chunk{b}{0}{n}}{\chunk{z}{0}{n - 1}} \{ (\pot[z_n] h) \varotimes \pot[z_n] \} \filt[\chunk{z}{0}{n}] h 
- \mumeas{\chunk{b}{0}{n}}{\chunk{z}{0}{n - 1}} \{ \pot[z_n] \varotimes (\pot[z_n] h) \} \filt[\chunk{z}{0}{n}] h \\
+ \mumeas{\chunk{b}{0}{n}}{\chunk{z}{0}{n - 1}} \pot[z_n]^{\varotimes 2} (\filt[\chunk{z}{0}{n}] h)^2
= \mumeas{\chunk{b}{0}{n}}{\chunk{z}{0}{n - 1}} \{Ê\pot[z_n] (h - \filt[\chunk{z}{0}{n}] h) \}^{\varotimes 2}, 
\end{multline*}
 as $\N$ tends to infinity. Now, applying the previous limit to \eqref{eq:filtering:variance:numerator} and using \eqref{eq:variance:alt:expression}, \eqref{eq:zero:identity}, and the second part of Lemma~\ref{lemma:mu:convergence} yields 
\begin{multline*}
\varestfilt[2]{\chunk{z}{0}{n}}[\lambda](h) = \frac{\varest[2]{\chunk{z}{0}{n - 1}}[\lambda](\pot[z_n] \{h - \filtpart[\chunk{z}{0}{n}] h\})}{(\predpart[\chunk{z}{0}{n - 1}] \pot[z_n])^2} \\Ê
\plim  \frac{\variance[2]{\chunk{z}{0}{n - 1}}[\lambda](\pot[z_n] \{ h - \filt[\chunk{z}{0}{n}] h \})}{(\pred[\chunk{z}{0}{n - 1}] \pot[z_n])^2} = \filtvariance[2]{\chunk{z}{0}{n}}[\lambda](h)
\end{multline*}
as $\N$ tends to infinity, which completes the proof. 
\end{proof}

\subsection{Proof of Theorem~\ref{thm:tightness:bias:filter}}

\begin{proof}
Pick arbitrarily $n \in \nset$, $\lag \in \nset$, and $h \in \bmf{\Xfd}$. Using the expression of $\filtvariance[2]{\chunk{\per}{0}{n}}(h)$Ê derived in the proof of \cite[Theorem~11]{douc:moulines:olsson:2014}, one may express the bias $\biasfilt{\chunk{\per}{0}{n}}{\lag}(h)$ as 
\begin{equation*} 
\biasfilt{\chunk{\per}{0}{n}}{\lag}(h) = \sum_{m = 0}^{\lagtime{n}{\lag} - 1} \int \pred[\chunk{\per}{0}{m - 1}](\rmd x) \left( \frac{\DDelta{\delta_x,\pred[\chunk{\per}{0}{m - 1}]}{\chunk{\per}{m}{n - 1}}{\pot[\per_n] h, \pot[\per_n]}}
{(\pred[\chunk{\per}{0}{m - 1}] \uk[\chunk{\per}{m}{n - 1}] \1_{\Xsp} )^2} \right)^2,  
\end{equation*}
where $\DDelta{\delta_x,\pred[\chunk{\per}{0}{m - 1}]}{\chunk{\per}{m}{n - 1}}{\pot[\per_n] h, \pot[\per_n]}$ is given by \eqref{eq:def-Delta}. Thus, the proof is finalised by following closely the lines of the proof of Theorem~\ref{thm:tightness:bias} and noting that the statement of \cite[Proposition~16(iii)]{douc:moulines:olsson:2014} still holds true when $h$Ê and $\1_\Xsp$ are replaced by $\pot[\per_n] h$ and $\pot[\per_n]$, respectively. 
\end{proof}

\subsection{Proof of Theorem~\ref{thm:strong:bias:bound}}
\label{sec:proof:strong:bias:bound}

\begin{proof}
If $n \leq \lag$, the bias vanishes by definition; we thus assume that $n > \lag$. Write, for $m \in \intvect{0}{n - \lag}$ and $x \in \Xsp$, 
\begin{multline} \label{eq:strong:bound:decomposition}
\frac{\uk[\chunk{z}{m}{n - 1}](h - \pred[\chunk{z}{0}{n - 1}] h)(x)}{\pred[\chunk{z}{0}{m - 1}] \uk[\chunk{z}{m}{n - 1}] \1_{\Xsp}} \\Ê
= \frac{\delta_x \uk[\chunk{z}{m}{n - 1}] \1_{\Xsp}}{\pred[\chunk{z}{0}{m - 1}] \uk[\chunk{z}{m}{n - 1}] \1_{\Xsp}} \left( \frac{\delta_x \uk[\chunk{z}{m}{n - 1}] h}{\delta_x \uk[\chunk{z}{m}{n - 1}] \1_{\Xsp}} - \frac{\pred[\chunk{z}{0}{m - 1}] \uk[\chunk{z}{m}{n - 1}] h}{\pred[\chunk{z}{0}{m - 1}] \uk[\chunk{z}{m}{n - 1}] \1_{\Xsp}} \right),   
\end{multline}
where \cite[Proposition~4.3.4]{delmoral:2004} bounds uniformly the second factor of \eqref{eq:strong:bound:decomposition} according to 
\begin{equation} \label{eq:strong:bound:second:term}
\left| \frac{\delta_x \uk[\chunk{z}{m}{n - 1}] h}{\delta_x \uk[\chunk{z}{m}{n - 1}] \1_{\Xsp}} - \frac{\pred[\chunk{z}{0}{m - 1}] \uk[\chunk{z}{m}{n - 1}] h}{\pred[\chunk{z}{0}{m - 1}] \uk[\chunk{z}{m}{n - 1}] \1_{\Xsp}} \right| \leq \varrho^{n - m} \| h \|_\infty. 
\end{equation}
To bound the first factor of \eqref{eq:strong:bound:decomposition}, note that
$$
\pred[\chunk{z}{0}{m - 1}] \uk[\chunk{z}{m}{n - 1}] \1_{\Xsp} = \filt[\chunk{z}{0}{m - 1}] \mk (\pot[z_m] \mk \uk[\chunk{z}{m + 1}{n - 1}] \1_{\Xsp}) \\
\geq \potlow \mdlow \mu \uk[\chunk{z}{m + 1}{n - 1}] \1_{\Xsp}
$$
and 
$$
\delta_x \uk[\chunk{z}{m}{n - 1}] \1_{\Xsp} = \pot[z_m](x) \mk \uk[\chunk{z}{m + 1}{n - 1}] \1_{\Xsp}(x) \leq \potup \mdup \mu \uk[\chunk{z}{m + 1}{n - 1}] \1_{\Xsp},
$$
where $(\mdlow, \mdup)$ and $(\potlow, \potup)$ are given in \B{ass:strong:mixing}(\ref{ass:strong:mixing-2}), implying that  
\begin{equation} \label{eq:strong:bound:first:term}
\frac{\delta_x \uk[\chunk{z}{m}{n - 1}] \1_{\Xsp}}{\pred[\chunk{z}{0}{m - 1}] \uk[\chunk{z}{m}{n - 1}] \1_{\Xsp}} \leq \frac{ \potup \mdup}{ \potlow \mdlow}.\end{equation}
Now, using \eqref{eq:strong:bound:second:term} and \eqref{eq:strong:bound:first:term}, proceed like  
\begin{multline*}
\bias{\chunk{z}{0}{n - 1}}{\lag}(h) = \sum_{m = 0}^{n - \lag - 1} \frac{\pred[\chunk{z}{0}{m - 1}] \{ \uk[\chunk{z}{m}{n - 1}](h - \pred[\chunk{z}{0}{n - 1}] h) \}^2 }{(\pred[\chunk{z}{0}{m - 1}] \uk[\chunk{z}{m}{n - 1}] \1_{\Xsp})^2} \\ 
\leq \| h \|_{\infty}^2 \left( \frac{ \potup \mdup}{ \potlow \mdlow} \right)^2  \sum_{m = 0}^{n - \lag - 1} \varrho^{2(n - m)} \leq c \| h \|_{\infty}^2 \varrho^{2(\lag + 1)}, 
\end{multline*}
with $c \eqdef (\potup \mdup / \potlow \mdlow)^2 / (1 - \varrho^2)$, and the proof is complete. 
\end{proof}

\bibliographystyle{plain}
\bibliography{mybibfile}

\end{document}